\documentclass[a4paper,USenglish,cleveref, nameinlink, autoref, thm-restate, numberwithinsect]{lipics-v2021}
\usepackage{graphicx} 
\usepackage{amsmath}
\usepackage{tikz}
\usetikzlibrary{decorations.pathreplacing,arrows,math,shapes}
\usetikzlibrary{decorations.pathmorphing}
\usepackage{amsthm}
\usepackage{bm}
\theoremstyle{definition}
\newtheorem{problem}{Problem}

\usetikzlibrary{positioning}

\title{A Graph Width Perspective on Partially Ordered Hamiltonian Paths}

\author{Jesse Beisegel}{Institute of Mathematics, Brandenburg University of Technology, Cottbus, Germany}{jesse.beisegel@b-tu.de}{https://orcid.org/0000-0002-8760-0169}{}
\author{Katharina Klost}{Institute of Computer Science, Freie Universität Berlin, Germany}{katharina.klost@fu-berlin.de}{https://orcid.org/0000-0002-9884-3297}{}
\author{Kristin Knorr}{Institute of Computer Science, Freie Universität Berlin, Germany}{kristin.knorr@fu-berlin.de}{https://orcid.org/0000-0003-4239-424X}{}
\author{Fabienne Ratajczak}{Institute of Mathematics, Brandenburg University of Technology, Cottbus, Germany}{fabienne.ratajczak@b-tu.de}{https://orcid.org/0000-0002-5823-1771}{}
\author{Robert Scheffler}{Institute of Mathematics, Brandenburg University of Technology, Cottbus, Germany}{robert.scheffler@b-tu.de}{https://orcid.org/0000-0001-6007-4202}{}

\authorrunning{J. Beisegel, K. Klost, K. Knorr, F. Ratajczak, and R. Scheffler}

\tikzstyle{vertex}=[draw, circle, fill=black, inner sep=1.5pt]

\renewcommand{\P}{\ensuremath{\mathsf{P}}}
\newcommand{\NP}{\ensuremath{\mathsf{NP}}}
\newcommand{\XP}{\ensuremath{\mathsf{XP}}{}}
\newcommand{\FPT}{\ensuremath{\mathsf{FPT}}}
\newcommand{\W}{\ensuremath{\mathsf{W[1]}}}
\newcommand{\MSO}{\ensuremath{\mathsf{MSO}}}
\newcommand{\cP}{\ensuremath{\mathcal{P}}}
\newcommand{\A}{\ensuremath{\mathcal{A}}}
\renewcommand{\O}{\ensuremath{\mathcal{O}}}
\newcommand{\G}{\ensuremath{\mathcal{G}}}
\newcommand{\T}{\ensuremath{\mathcal{T}}}

\newcommand{\N}{\ensuremath{\mathbb{N}}}
\newcommand{\R}{\ensuremath{\mathcal{R}}}

\newcommand{\pip}{{\pi'}}

\renewcommand{\phi}{\varphi}

\newcommand{\paramfont}{\slshape}
\newcommand{\param}[1]{{\paramfont #1}}

\crefname{claim}{Claim}{Claims}
\Crefname{claim}{Claim}{Claims}

\theoremstyle{claimstyle}

\newcommand{\pef}[1]{(P\ref{#1})}



\theoremstyle{plain}

\theoremstyle{definition}

\theoremstyle{claimstyle}

\newenvironment{theorem*}
 {\expandafter\def\expandafter\thetheorem\expandafter{\thetheorem}\theorem}
 {\endtheorem}

\newenvironment{lemma*}
  {\expandafter\def\expandafter\thelemma\expandafter{\thelemma}\lemma}
 {\endlemma}

\newenvironment{observation*}
  {\expandafter\def\expandafter\theobservation\expandafter{\theobservation}\observation}
 {\endobservation}

\newenvironment{definition*}
  {\expandafter\def\expandafter\thedefinition\expandafter{\thedefinition}\definition}
 {\enddefinition}

\newenvironment{corollary*}
  {\expandafter\def\expandafter\thecorollary\expandafter{\thecorollary}\corollary}
 {\endcorollary}

\renewenvironment{claim*}
  {\expandafter\def\expandafter\theclaim\expandafter{\theclaim}\claim}
 {\endclaim}

\newenvironment{problem*}{\expandafter\def\expandafter\theproblem\expandafter{\theproblem}\problem}
 {\endproblem}

 \crefname{claim*}{Claim}{Claims}
\Crefname{claim*}{Claim}{Claims}

\nolinenumbers
\hideLIPIcs
\begin{document}

\maketitle

\begin{abstract}
We consider the problem of finding a Hamiltonian path with precedence constraints in the form of a partial order on the vertex set. This problem is known as \textsc{Partially Ordered Hamiltonian Path Problem} (POHPP). Here, we study the complexity for graph width parameters for which the ordinary \textsc{Hamiltonian Path} problem is in \FPT. We show that  POHPP is \NP-complete for graphs of pathwidth~4. We complement this result by giving polynomial-time algorithms for graphs of pathwidth~3 and treewidth~2. Furthermore, we show that POHPP is \NP-hard for graphs of clique cover number~2 and \W-hard for some distance-to-\G~parameters, including distance to path and distance to clique. In addition, we present \XP{} and \FPT{} algorithms for parameters such as distance to block and feedback edge set number.
\end{abstract}

\section{Introduction}

\subparagraph*{Hamiltonian Paths and Cycles with Precedence} 

For some applications of the well-known \textsc{Traveling Salesman Problem} (TSP) it is necessary to add additional precedence constraints to the vertices which ensure that some vertices are visited before others in a tour. Examples are \emph{Pick-up and Delivery}  Problems~\cite{parragh2008survey,parragh2008survey2} or the \emph{Dial-a-Ride} problem~\cite{psaraftis1980dynamic}, where goods or people have to be picked up before they can be brought to their destination.

For that reason, both the cycle and the path variant of the TSP have been considered together with precedence constraints. The cycle variant is known as \textsc{Traveling Salesman Problem with Precedence Constraints} (TSP-PC) and has been studied, e.g., in \cite{ahmed2001travelling,bianco1994exact}. The path variant, known as the \textsc{Sequential Ordering Problem} (SOP) or the \textsc{Minimum Setup Scheduling Problem}, has been studied, e.g., in~\cite{ascheuer1993cutting,colbourn1985minimizing,escudero1988inexact,escudero1988implementation}. 

Of course, all these problems are \NP-complete and research in these topics has mainly been focused on heuristic algorithms and integer-programming approaches. Furthermore, these problems are defined over complete graphs with an additional cost function. The unweighted variants \textsc{Hamiltonian Path} and \textsc{Hamiltonian Cycle} with precedence constraints for non-complete graphs have not received the same level of attention for a long time. Results have been only given for the very restricted variants where one or both endpoints of the Hamiltonian path are fixed. For these problems, polynomial-time algorithms have been presented for several graph classes including (proper) interval graphs~\cite{asdre2010fixed,asdre2010polynomial,li2017linear,mertzios2010optimal}, distance-hereditary graphs~\cite{hsieh2004efficient,yeh1998path}, and rectangular grid graphs~\cite{itai1982hamilton}.

To overcome this lack of research, Beisegel et al.~\cite{beisegel2024computing} introduced the following problem.

\begin{problem}{\textsc{Partially Ordered Hamiltonian Path Problem} (POHPP)}
\begin{description}
\item[\textbf{Instance:}] A graph $G$, a partial order $\pi$ on the vertex set of $G$.
\item[\textbf{Question:}]
Is there an ordered Hamiltonian path $(v_1, \dots, v_n)$ in $G$ such that for all $i, j \in \{1,\dots,n\}$ it holds that if $(v_i,v_j) \in \pi$, then $i \leq j$?
 \end{description}
\end{problem}

They also introduced the edge-weighted variant \textsc{Minimum Partially Ordered Hamiltonian Path Problem} (MinPOHPP). The authors show that  POHPP is already \NP-hard for complete bipartite graphs and complete split graphs -- graph classes where \textsc{Hamiltonian Path} is trivial.
They also show that  POHPP is \W-hard when parameterized by the \param{width} of the partial order, i.e., the largest number of pairwise incomparable elements. Furthermore, they show that the \XP{} algorithm for that parametrization presented in the 1980s by Coulbourn and Pulleyblank~\cite{colbourn1985minimizing} is asymptotically optimal -- assuming the Exponential Time Hypothesis (ETH). They improve the algorithm to \FPT{} time if the problem is parameterized by the partial order's \param{distance to linear order}. Finally, the authors present a polynomial-time algorithm for MinPOHPP on outerplanar graphs.

\subparagraph*{Graph Width Parameters} 

Since many graph problems are \NP-hard, legions of researchers have been trying to find tractable instances of these problems. One approach is the idea to consider \emph{graph classes}, i.e., subsets of the set of graphs. Another approach are \emph{graph width parameters}. In essence, such a parameter is a mapping from the set of graphs to the integers. The idea is that graphs mapped to large values are in some sense more complex than graphs mapped to small values. Graph width parameters can also been seen as infinite families of graph classes since for every integer $k$ the graphs of parameter value $k$ form a graph class.

One of the first graph width parameters considered to solve \NP-hard problems was \param{bandwidth}~\cite{monien1980bounding}. Later, \param{treewidth} -- probably the most-famous graph width parameter -- gained much attention. Many problems that are \NP-hard for general graphs can be solved in polynomial time when the \param{treewidth} is bounded. In particular, Courcelle~\cite{courcelle1990monadic,courcelle1992monadic} showed in his famous theorem that every problem expressible in \emph{monadic second order logic} ($\MSO_2$) is solvable in \FPT{} time when parameterized by \param{treewidth}. One disadvantage of \param{treewidth} is the fact that the parameter is \emph{sparse}, i.e., the number of edges of a graph with $n$ vertices and \param{treewidth}~$k$ is bounded by $\O(kn)$. Nevertheless, there are graph classes with many edges which allow efficient algorithms. To overcome this problem, \param{cliquewidth} has been considered, a parameter that generalizes \param{treewidth}. Courcelle et al.~\cite{courcelle2000linear} showed for \param{cliquewidth} that every problem expressible in $\MSO_1$ is solvable in \FPT{} time. In this logic, quantification is only allowed over sets of vertices while in $\MSO_2$ one can also quantify over sets of edges. Besides \param{treewidth} and \param{cliquewidth}, several other width parameters have been introduced as is demonstrated in a survey from 2008~\cite{hlineny2008width}. Since then, research on this topic has continued intensively and led to parameters such as \param{mim-width}~\cite{vatshelle2012new}, \param{twin-width}~\cite{bonnet2022twin-width}, and the \param{tree-independence number}~\cite{dallard2024treewidth}.

In many cases, graph width parameters are strongly related to certain graph classes. If there is a graph class where many problems can be solved efficiently, then it is a common approach to introduce parameters that describe how far away a graph is from this graph class. In some sense, \param{treewidth} describes how tree-like a graph is. An approach used for graph classes that can be defined by intersection models is to generalize these models to a whole hierarchy. This approach was, e.g., used for interval graphs~\cite{beisegel2024simultaneous,fellows2009parameterized,francis2015maximum,jiang20212parameterized}. Another way to define such graph width parameters is to consider the smallest number of vertices or edges that have to be removed such that the resulting graph is in the class. We refer to these parameters as \param{(vertex) distance to $\G$} and \param{edge distance to $\G$} where $\G$ is the respective graph class.\footnote{Alternative terms for these parameters are \param{$\G$ vertex deletion number} and \param{$\G$ edge deletion number}.} Some of these parameters have their own names. The \param{distance to edgeless} is called the \param{vertex cover number}. The \param{distance to forest} is called \param{feedback vertex set number} while the \param{edge distance to forest} is called \param{feedback edge set number} or \param{circuit rank}. Here, we will also consider a special variant of these distance parameters that is motivated by the parameter \param{twin-cover number} introduced by Ganian~\cite{ganian2011twin-cover}. We say a graph has \param{distance to $\G$ module(s)}\footnote{If the class $\G$ contains only connected graphs, then we use \param{distance to $\G$ module} otherwise we use \param{distance to $\G$ modules}.} if there is a set $W \subseteq V(G)$ of $k$ vertices such that $G - W$ is in $\G$ and every component of $G - W$ forms a module in $G$. Using this terminology, the \param {twin-cover number} is equal to the \param{distance to cluster modules}.

\subparagraph*{Intersections}

\textsc{Hamiltonian Path} and \textsc{Hamiltonian Cycle} belong to the most-famous \NP-hard graph problems and, hence, have been studied for a wide range of graph classes and graph width parameters.\footnote{In general, \textsc{Hamiltonian Path} seems to lead a shadowy existence since many positive and negative algorithmic results are only given for its more popular sibling \textsc{Hamiltonian Cycle}.}

One of the first results for width parameters were polynomial-time algorithms for \textsc{Hamiltonian Cycle} and TSP on graphs of bounded \param{bandwidth}~\cite{lawler1985traveling,monien1980bounding}. Since both \textsc{Hamiltonian Cycle} and \textsc{Hamiltonian Path} are expressible in $\MSO_2$, Courcelle's theorem~\cite{courcelle1990monadic,courcelle1992monadic} implies \FPT{} algorithms for these problems when parameterized by \param{treewidth}. The running time bounds depending on the \param{treewidth}~$k$ given by Courcelle's theorem are quite bad. However, there are also \FPT{} algorithms with single exponential dependency on $k$~\cite{cygan2022solving}. Probably, these results cannot be transfered to \param{cliquewidth}. In contrast to $\MSO_2$, both \textsc{Hamiltonian Cycle} and \textsc{Hamiltonian Path} are not expressible in $\MSO_1$~\cite[Corollary~5.3.5]{ebbinghaus1995finite}. In fact, \textsc{Hamiltonian Cycle} is \W-hard when parameterized by \param{cliquewidth}~\cite{fomin2010intractability} and -- unless $\FPT = \W$ -- it is not solvable in \FPT{} time. Furthermore, it was shown that -- unless the ETH fails -- the problem cannot be solved in $f(k) n^{o(k)}$ time where $k$ is the \param{cliquewidth} and $n$ is the number of vertices~\cite{fomin2018clique-width}. This lower bound is matched by an $f(k) n^{\O(k)}$ algorithm given by Bergougnoux et al.~\cite{bergougnoux2020optimal}. \XP{} algorithms for \param{cliquewidth} with worse exponents have been presented earlier by Wanke~\cite{wanke1994k-nlc} (for the equivalent NLC-width) and by Espelage et al.~\cite{espelage2001solve}. 

As \FPT{} algorithms for \textsc{Hamiltonian Path} and \textsc{Hamiltonian Cycle} parameterized by \param{cliquewidth} are highly unlikely, the research focused on upper bounds of \param{cliquewidth}. \FPT{} algorithms were presented for \param{neighborhood diversity}~\cite{lampis2012algorithmic}, \param{distance to cluster}~\cite{doucha2012cluster,jansen2013power}, \param{modular width}~\cite{gajarsky2013parameterized} and \param{split matching width}~\cite{saether2016between,saether2017solving}. Besides this, parameterized algorithms have also been given for graph width parameters incomparable to \param{cliquewidth} and \param{treewidth}. Examples are \FPT{} algorithms for \param{distance to proper interval graphs}~\cite{golovach2020graph} and for two parameters that describe the distance to Dirac's condition on the existence of Hamiltonian paths~\cite{jansen2019hamiltonicity}. Most recently, \XP{} and \FPT{} algorithms for the \param{independence number} have been presented~\cite{fomin2024hamiltonicity,jedlickova2024hamiltonian}.

\subparagraph*{Our Contribution} 

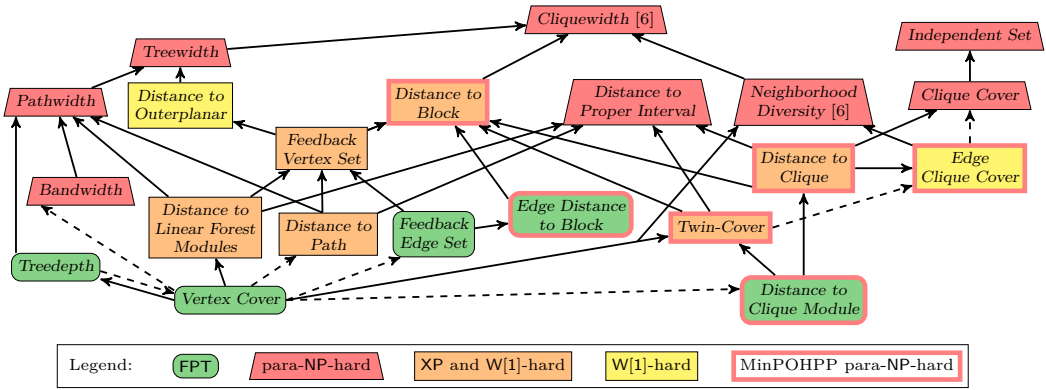
\begin{figure}
\centering
\resizebox{\textwidth}{!}{
\begin{tikzpicture}[xscale=2.5,yscale=1]
  \tikzset{np/.style={draw, trapezium, trapezium angle=77.5, fill=red!50!white}}
  \tikzset{fpt/.style={draw,rectangle, fill=green!30!lightgray,rounded corners,
  }}
  \tikzset{xpw/.style={draw,rectangle,fill=orange!50!white}}
  \tikzset{xp/.style={draw,rectangle,fill=yellow!20!white}}
  \tikzset{w/.style={draw,rectangle,fill=yellow!70!white}}
  \tikzset{fptnp/.style={draw,rectangle, fill=green!30!lightgray,  draw=red!50!white, line width=2pt,
  rounded corners
  }}
  \tikzset{minnp/.style={draw=red!50!white, line width=2pt}}
  \tikzset{xpnp/.style={draw,rectangle,fill=yellow!20!white, draw=red!50!white, line width=2pt, 
  }}
  \tikzset{wnp/.style={draw,rectangle,fill=yellow!70!white,  draw=red!50!white, line width=2pt, 
  }} 
  \tikzset{xpwnp/.style={draw,rectangle,fill=orange!50!white,  draw=red!50!white, line width=2pt,
  }}
  \tikzset{open/.style={draw,rectangle}}
  \scriptsize
  \paramfont
  \node[fpt] (vc) at (-0.95,0) {\mathstrut Vertex Cover};
  \node[xpw, align=center] (d2path) at (-0.36,1) {\mathstrut Distance to \\ Path};
  \node[xpwnp, align=center] (d2clique) at (2.5,2) {\mathstrut Distance to \\ Clique};
  \node[xpwnp, align=center] (d2block) at (0.3,3) {\mathstrut Distance to \\ Block};
  \node[fptnp, align=center] (ed2block) at (1.1,1.3) {\mathstrut Edge Distance \\ to Block};
  \node[xpwnp, align=center] (tc) at (2,1.1) {\mathstrut Twin-Cover};
  \node[wnp, align=center] (ecc) at (3.5,2) {\mathstrut Edge \\ Clique Cover};
  \node[np, align=center] (nd) at (2.5,3) {\mathstrut Neighborhood \\ Diversity~{\normalfont \cite{beisegel2024computing}}};
  \node[np, align=center] (cc) at (3.5,3.1) {\mathstrut Clique Cover};
  \node[np, align=center] (is) at (3.5,4) {\mathstrut Independent Set};
  \node[fpt] (td) at (-2,0.5) {\mathstrut Treedepth};
  \node[fpt, align=center] (fes) at (0.275,1) {\mathstrut Feedback \\ Edge Set};
  \node[np] (bw) at (-1.85,1.65) {\mathstrut Bandwidth};
  \node[np] (pw) at (-2,3) {\mathstrut Pathwidth};
  \node[np] (tw) at (-1.25,3.75) {\mathstrut Treewidth};
  \node[xpw, align=center] (fvs) at (-0.4,2.3) {\mathstrut Feedback \\ Vertex Set};
  \node[w, align=center] (d2op) at (-1.25,2.95) {\mathstrut Distance to \\ Outerplanar};
  \node[np, align=center] (d2pi) at (1.5,3) {\mathstrut Distance to \\ Proper Interval};
  \node[np, align=center] (cw) at (1.25,4.25) {\mathstrut Cliquewidth~{\normalfont \cite{beisegel2024computing}}};
  \node[fptnp, align=center] (d2cm) at (2.5,0) {\mathstrut Distance to \\ Clique Module};
  \node[xpw, align=center] (d2lfm) at (-1.1,1.1) {\mathstrut Distance to \\ Linear Forest \\ Modules};
   \draw[thick,-stealth'] (bw) -- (pw);
   \draw[thick,-stealth'] (pw) -- (tw);
   \draw[thick,-stealth'] (vc.0) -- (tc);
   \draw[thick,-stealth'] (vc.180) -- (td.345);
   \draw[thick,-stealth'] (td.160) -- (pw.-160);
   \draw[thick,-stealth'] (d2clique) -- (cc);
   \draw[thick,-stealth'] (d2clique) -- (ecc);
   \draw[thick,-stealth'] (d2op) -- (tw);
   \draw[thick,-stealth'] (ecc) -- (nd);
   \draw[thick,-stealth'] (cc) -- (is);
   \draw[thick,-stealth'] (d2block) -- (cw);
   \draw[thick,-stealth'] (tw) -- (cw);
   \draw[thick,-stealth'] (nd) -- (cw);
   \draw[thick,-stealth'] (fes) -- (ed2block);
   \draw[thick,-stealth'] (ed2block.166) -- (d2block);
   \draw[thick,-stealth'] (d2cm) -- (d2clique);
   \draw[thick,-stealth'] (d2cm) -- (tc);
   \draw[thick,-stealth'] (d2path) -- (d2pi);
   \draw[thick,dashed,-stealth'] (td.355) -- (vc.173);
   \draw[thick,dashed,-stealth'] (vc.0) -- (d2cm.170);
   \draw[thick,dashed,-stealth'] (vc.170) -- (bw.200);
   \draw[thick,dashed,-stealth'] (vc.0) -- (fes.215);
   \draw[thick,dashed,-stealth'] (ecc) -- (cc);

   \draw[thick, -stealth'] (d2path.105) -- (fvs);
   \draw[thick, -stealth'] (fvs) -- (d2block);
   \draw[thick, -stealth'] (d2lfm) -- (fvs);
   \draw[thick, -stealth'] (d2lfm) -- (d2pi);
   \draw[thick, -stealth'] (fvs) -- (d2block);
   \draw[thick, -stealth'] (tc) -- (d2block.-30);
   \draw[thick, dashed, -stealth'] (tc.0) -- (ecc);
   \draw[thick, -stealth'] (d2clique.200) -- (d2block.-20);
   \draw[thick, -stealth'] (d2lfm) -- (pw);
   \draw[thick, -stealth'] (tc) -- (d2pi);
   \draw[thick, -stealth'] (d2clique) -- (d2pi);
   \draw[thick, -stealth'] (fes) -- (fvs);
   \draw[thick, -stealth'] (d2path.105) -- (pw);
   \draw[thick, -stealth'] (fvs) -- (d2op);
   \draw[thick, dashed, -stealth'] (vc) -- (d2path);
   \draw[thick, -stealth'] (vc) -- (d2lfm);
   \draw[thick, -stealth'] (1.5,0.89) -- (nd.200);

   \node[draw=black] at (0.75,-1) {\begin{tikzpicture}
         \normalfont
       \node (Legend) at (-2,-1) {Legend:};
       \node[fpt] (FPT) [right=0.5cm of Legend] {\FPT};
       \node[np] (NP) [right=0.5cm of FPT] {para-\NP-hard};
       \node[xpw] (XPW) [right=0.5cm of NP] {\XP{} and \W-hard};
       \node[w] (W) [right=0.5cm of XPW] {\W-hard};
       \node[minnp] (MINNP) [right=0.5cm of W] {MinPOHPP para-\NP-hard};
   \end{tikzpicture}};
   
 \end{tikzpicture}
}
\caption{Diagram illustrating the complexity results for (Min)POHPP for different graph width parameters. A directed solid edge from parameter $P$ to parameter $Q$ means that a bounded value of $P$ implies a bounded value for $Q$. A directed dashed edge implies that this relation does not hold in general but for traceable graphs, i.e., graphs having a Hamiltonian path. If a directed solid path from $P$ to $Q$ is missing, then parameter $Q$ is unbounded for the graphs of bounded $P$. The same holds for the traceable graphs if there is also no path using dashed edges.}\label{fig:parameters}
\end{figure}

So far, there seems to be no research on the parameterized complexity of \textsc{Hamiltonian Path} with precedence constraints in the context of graph width parameters. To change this, we study the complexity of (Min)POHPP for graph width parameters where \textsc{Hamiltonian Path} can be solved in \FPT{} time (see \cref{fig:parameters} for an overview). It follows directly from the \NP-completeness of  POHPP on complete bipartite graphs~\cite{beisegel2024computing} that the problem is para-\NP-hard for several parameters for which \textsc{Hamiltonian Path} can be solved in \FPT{} time or \XP{} time.

\begin{observation*}
     POHPP is \NP-complete on graphs of \param{cliquewidth}~$2$, \param{neighborhood diversity}~$2$, \param{modular width}~$2$ and \param{split matching width}~$1$.
\end{observation*}

Note that graphs of \param{cliquewidth}~1 and \param{modular width}~1 are edgeless while graphs of \param{neighborhood diversity}~1 are either edgeless or cliques. So in all these cases,  POHPP is trivial.

In \cref{sec:treewidth} we will extend this observation to several \param{treewidth}-like parameters. To this end, we show that  POHPP is \NP-complete on rectangular grid graphs of height~7 and on proper interval graphs of clique number~5. The latter implies that  POHPP is \NP-complete for graphs of \param{bandwidth}, \param{pathwidth}, and \param{treewidth}~4. We complement these results by showing that  MinPOHPP can be solved in polynomial time for graphs of \param{pathwidth}~3 and of \param{treewidth}~2, leaving the case of \param{treewidth}~3 open. Furthermore, we present a simple argument why  MinPOHPP is in \FPT{} for \param{treedepth}.

\Cref{sec:d2p} is dedicated to \param{distance to $\G$} parameters where $\G$ is a sparse graph class. We show that  POHPP is \W-hard for \param{distance to path} and \param{distance to linear forest modules}. For MinPOHPP, we give an \FPT{} algorithm for \param{feedback edge set number} and an \XP{} algorithm for \param{feedback vertex set number}. We also give an \FPT{} algorithm for planar graphs parameterized by the number of vertices that lie in the interior of the outer face.

In \cref{sec:d2c}, we consider graph width parameters that are bounded for cliques. We show that  POHPP is \NP-complete for graphs of \param{clique cover number}~2. Furthermore, we prove \W-hardness for \param{distance to clique} and \param{distance to cluster modules} aka \param{twin cover number}. On the positive side, we present an \XP{} algorithm for POHPP when parameterized by \param{distance to block} and \FPT{} algorithms when parameterized by \param{edge distance to block} or \param{distance to clique module}.

\section{Preliminaries}

\subparagraph{General Notation and Partial Orders} A \emph{partial order} $ \pi $ on a set $X$ is a reflexive, antisymmetric and transitive relation on $X$. The tuple $(X, \pi)$ is then called a \emph{partially ordered set}. We also denote $(x,y) \in \pi$ by $x \prec_\pi y$ if $x \neq y$. If it is clear which partial order is meant, then we sometimes omit the index. A \emph{minimal element} of a partial order $\pi$ on $X$ is an element $x \in X$ for which there is no element $y \in X$ with $y \prec_\pi x$. The \emph{reflexive and transitive closure} of a relation $\R$ is the smallest relation $\R'$ such that $\R \subseteq \R'$ and $\R'$ is reflexive and transitive.

A \emph{linear ordering} of a finite set $X$ is a bijection $\sigma: X \rightarrow \{1,2,\dots,|X|\}$. We will often refer to linear orderings simply as orderings. Furthermore, we will denote an ordering by a tuple $(x_1, \ldots, x_n)$ which means that $\sigma(x_i) = i$. Given two elements $x$ and $y$ in $X$, we say that $x$ is \emph{to the left} (resp. \emph{to the right}) of $y$ if $\sigma(x)<\sigma(y)$ (resp. $\sigma(x)>\sigma(y)$) and we denote this by $x \prec_{\sigma}y$ (resp.  $x \succ_{\sigma}y$). A \emph{linear extension} of a partial order $\pi$ is a linear ordering $\sigma$ of $X$ that fulfills all conditions of $\pi$, i.e., if $x \prec_\pi y$, then $x \prec_\sigma y$.

For $n \in \N$, the notation $[n]$ refers to the set $\{i \in \N \mid 1 \leq i \leq n\}\).

\subparagraph{Graphs and Graph Classes}
All the graphs considered here are finite. For the standard notation of graphs we refer to the book of Diestel~\cite{diestel}.

A vertex $v$ of a connected graph $G$ is a \emph{cut vertex} if $G - v$ is not connected. If $G$ does not contain a cut vertex, then $G$ is \emph{2-connected}. A \emph{block} of a graph is an inclusion-maximal 2-connected induced subgraph. The \emph{block-cut tree} $\mathcal{T}$ of $G$ is the bipartite graph that contains a vertex for every cut vertex of $G$ and a vertex for every block of $G$ and the vertex of block $B$ is adjacent to the vertex of a cut vertex $v$ in $\mathcal{T}$ if $B$ contains $v$.

A \emph{Hamiltonian path} of a graph $G$ is a path that contains all the vertices of $G$. A graph is \emph{traceable} if it has a Hamiltonian path. Here, we only consider \emph{ordered} Hamiltonian paths, i.e., one of the two possible orderings of the path is fixed. Given a partial order $\pi$ on a graph's vertex set, an ordered Hamiltonian path is a \emph{$\pi$-extending Hamiltonian path} if its order is a linear extension of $\pi$.

A graph is a \emph{linear forest} if all its connected components are paths. A graph is a \emph{cluster graph} if all its components are cliques. A graph is a \emph{block graph} if its blocks are cliques. A graph $G$ is a \emph{proper interval graph} if it has a vertex ordering $(v_1, \dots, v_n)$ such that for all edges $v_iv_k$ and every $j \in \{i, \dots, k\}$ it holds that $v_iv_j$ and $v_jv_k$ are in $E(G)$. We call such an ordering a \emph{proper interval ordering}.

A graph is \emph{planar} if it has a crossing-free embedding in the plane, and together with this embedding it is called a \emph{plane graph}. For a plane graph $G$ we call the regions of $\mathbb{R}^2 \setminus G$ the \emph{faces} of $G$. Every plane graph has exactly one unbounded face which is called the \emph{outer face}. A graph is called \emph{outerplanar} if it has a crossing-free embedding such that all of the vertices belong to the outer face and such an embedding is also called \emph{outerplanar}. A \(h \times w\) grid graph is a graph with vertex set \(\{1,\dots h\}\times \{1,\dots w\}\) and edges between each pair of vertices with hamming distance one. We call $w$ the \emph{width} of the graph and $h$ the \emph{height} of the graph.

\subparagraph{Graph Width Parameters}
A \emph{tree decomposition} of a graph $G$ is a pair $(T,\{X_t\}_{t\in V(T)})$ consisting of a tree $T$ and a mapping assigning to each node $t\in V(T)$ a set $X_t\subseteq V(G)$ (called \emph{bag}) such that the union of all the bags equals $V(G)$, for every edge $uv\in E(G)$ there exists a bag $X_t$ such that $u,v\in X_t$, and for every vertex $v\in V(G)$ the bags containing $v$ form a subtree of $T$. The \emph{width} of a tree decomposition is the maximum size of a bag minus~1. The \param{treewidth} of a graph $G$ is the minimal width of a tree decomposition of $G$. The terms \emph{path decomposition} and \param{pathwidth} are defined accordingly, where the tree $T$ is replaced by a path. We then simply use the ordering $X_1, \dots, X_k$  of the bags in that path to describe the decomposition. A \emph{treedepth decomposition} of a graph consists of a rooted forest $F$ where every vertex of $G$ is mapped to a vertex in $F$ such that two vertices that are adjacent in $G$ have to form ancestor and descendant in $F$. The \param{treedepth} of a graph $G$ is the minimal height of a treedepth decomposition.

The \param{bandwidth} of a graph $G$ is the smallest integer $k$ such that there exists a vertex ordering $\sigma$ of $G$ where $|\sigma(u) - \sigma(v)| \leq k$ for every edge $uv \in E(G)$. A \emph{clique cover} of a graph $G$ is a partition of $V(G)$ into cliques. The \param{clique cover number} of $G$ is the minimal size of a clique cover of $G$.

Let $\G$ be a graph class. The \param{distance to $\G$} of a graph $G$ is the minimal number of vertices that need to be removed from $G$ so that the resulting graph belongs to $\G$. Similarly, the \param{edge distance to $\G$} of a graph $G$ is the minimal number of edges that need to be removed from $G$ so that the resulting graph belongs to $\G$. A \emph{module} of a graph is a subset $M \subseteq V(G)$ such that  for all vertices $u, v \in M$ it holds that every neighbor of $u$ outside of $M$ is also a neighbor of $v$. The \param{distance to $\G$ module(s)} is the minimal number of vertices that need to be removed from $G$ so that the resulting graph is in $\G$ and every component of the resulting graph is a module in $G$. Several of the distance parameters have got their own name. So \param{vertex cover number} is equivalent to \param{distance to edgeless}, \param{feedback vertex (edge) set number} is equivalent to \param{(edge) distance to tree} and \param{twin-cover number}~\cite{ganian2011twin-cover} is equivalent to \param{distance to cluster modules}.

\subparagraph{Complexity} We will use the following two problems in reductions.

\begin{problem*}{\textsc{3-Satisfiability} (3-SAT)}
\begin{description}
\item[\textbf{Instance:}] A boolean formula $\Phi = C_1 \land \dots \land C_m$ with $C_i = \ell_i^1 \lor \ell_i^2 \lor \ell_i^3$ where $\ell_i^j$ is some $x_j$ or some $\overline{x}_j$.
\item[\textbf{Question:}]
Is there a fulfilling assignment of $\Phi$?
 \end{description}
\end{problem*}

It is well-known that the 3-SAT problem is \NP-complete~\cite{karp1972reducibility}.

\begin{problem*}{\textsc{Multicolored Clique Problem} (MCP)}
\begin{description}
\item[\textbf{Instance:}] A graph $G$ with a proper coloring by $k$ colors.
\item[\textbf{Question:}]
Is there a clique $C$ in $G$ such that $C$ contains exactly one vertex of each color?
 \end{description}
\end{problem*}

The MCP was shown to be \W-hard by Pietrzak~\cite{pietrzak2003parameterized} and independently by Fellows et al.~\cite{fellows2009parameterized}. Note that this also holds if all color classes have the same size.

\section{Bandwidth, Pathwidth, Treewidth, and Treedepth}\label{sec:treewidth}
In this section we consider the computational complexity of (Min)POHPP when parameterized by \param{bandwidth}, \param{pathwidth}, \param{treewidth} and \param{treedepth}. 
On the hardness side of things, we show that POHPP and thus also MinPOHPP is \NP{}-complete for \param{bandwidth}, \param{pathwidth} and \param{treewidth} at least \(4\).
We contrast these hardness results with polynomial time algorithms for \param{pathwith} at most 3 and \param{treewidth} at most 2.
For \param{treedepth} at most \(k\), we show that a double exponential time algorithm in \(k\), independent of \(n\) exists.
\subsection{Hardness}
In this section we show that  POHPP is \(\NP\)-complete for graphs of \param{bandwidth}, \param{pathwidth} and \param{treewidth} at most \(4\). 
This is shown, by showing \NP-completeness for proper interval graphs of clique number \(5\) (\cref{thm:np-interval}).
One might wonder if there are more structured graph classes with bounded \param{bandwidth} for which the problem becomes easier. 
For a very structured subclass, namely \(h\times w\) grid graphs, we show in \cref{thm:np-grid} that the problem is \NP-complete for \(\min\{h,w\}\geq 7\).

\begin{theorem}\label{thm:np-interval}
     POHPP is \NP-complete on proper interval graphs of clique number~5.
\end{theorem}

\begin{proof}
We present a reduction from 3-SAT to  POHPP on proper interval graphs of clique number~5. Let $\Phi$ be a 3-SAT formula with the variables $x_1, \dots, x_n$ and the clauses $c_1, \dots c_m$. Let $\ell_i^1$, $\ell_i^2$, and $\ell_i^3$ be the literals of $c_i$. Note that these literals may be negated. 

In the following we construct an instance $(G,\pi)$ of  POHPP (see \cref{fig:unit-interval}). We use the following subgraphs for $G$.

\begin{description}
    \item[Start Gadget] The start gadget $S$ consists of the two adjacent vertices $s$ and $s'$.
    \item[Variable Gadget] For each $x_1, \dots, x_n$, there is a variable gadget $X_i$ that consists of the adjacent vertices $x_i$ and $\overline{x}_i$. We call these vertices \emph{variable vertices}.
    \item[Clause Gadget] For each clause $c_i$, there is a clause gadget $C_i$ consisting of the vertices $a^1_i$, $a^2_i$, $\ell^1_i$, $\ell^2_i$, $\ell^3_i$, $b^1_i$, and $b^2_i$. The vertices $\ell^1_i$, $\ell^2_i$ and $\ell^3_i$ form a clique. Similarly $a^1_i$ and $a^2_i$ as well as $b^1_i$ and $b^2_i$ form cliques of size~2, respectively. The vertices $a^1_i$ and $b^1_i$ are adjacent to all of $\ell^1_i$, $\ell^2_i$, and $\ell^3_i$. Vertex $a^2_i$ is only adjacent to $\ell^2_i$ and $\ell^3_i$, while $b^2_i$ is only adjacent to $\ell^1_i$ and $\ell^3_i$ (see two examples given in the rounded-corner boxes in \cref{fig:unit-interval}).
    \item[End Gadget] The end gadget $T$ consists of the two adjacent vertices $t$ and $t'$.
    \item[Backbone] The backbone $B$ is a path consisting of the vertices $\{r_i \mid 0 \leq i \leq n\}$ and $\{u_i,v_i,w_i \mid 1 \leq i \leq m\}$ in the following order: $(r_0, r_1 \dots, r_n, u_1, v_1, w_1, \dots, u_n, v_n, w_n)$.
\end{description}

To complete the construction of $G$, we explain how these gadgets are combined. We order the gadgets in the following way $S, X_1, \dots, X_n, C_1, \dots, C_m, T$. For every gadget, we define \emph{entry vertices} and \emph{exit vertices}. For start, variable and end gadgets, all vertices are entry and exit vertices. For clause gadgets, $a_i^1$ and $a_i^2$ are the entry vertices, while $b_i^1$ and $b_i^2$ are the exit vertices. The exit vertices of a gadget are completely adjacent to the entry vertices of the succeeding gadget. The backbone is connected to all other gadgets in the following way. Vertex $r_0$ is adjacent to $s$, $s'$, $x_1^1$ and $\overline{x}_1$. Vertex $r_i$ with $1 \leq i < n$ is adjacent to $x_i$, $\overline{x}_i$, $x_{i+1}$, and $\overline{x}_{i+1}$. The vertex $r_n$ is adjacent to $x_n$, $\overline{x}_n$ as well as $a_1^1$ and $a_1^2$. Vertex $u_i$ is adjacent to $a_i^1$, $a_i^2$ as well as $\ell_i^1$, $\ell_i^2$, and $\ell_i^3$. Vertex $w_i$ is adjacent to $\ell_i^1$, $\ell_i^2$, and $\ell_i^3$ as well as $b_i^1$ and $b_i^2$. Vertex $w_i$ is adjacent to $b_i^1$, $b_i^2$, $a_{i+1}^1$, and $a_{i+1}^2$ except from $w_n$ which is adjacent to $b_n^1$, $b_n^2$ as well as $t$ and $t'$.

To prove that $G$ is a proper interval graph, we consider the following vertex ordering:
\begin{align*}
\sigma = (s, s', r_0, x_1, \overline{x}_1, r_1, x_2, \dots, x_n, \overline{x}_n, r_n, a_1^2, a_1^1, u_1, \ell_1^2, \ell_1^3, \ell_1^1, v_1, b_1^1, b_1^2, w_1, a_2^2, a_2^1,\\ \dots, b_m^1, b_m^2, w_m, t, t').
\end{align*}
It can easily be checked that this ordering is a proper interval ordering. Furthermore, the \param{bandwidth} of $\sigma$ is 4. Since the \param{clique number} of a proper interval graph is equal to its \param{bandwidth} plus~1~\cite[Theorem~4.1]{kaplan1996pathwidth}, we know that the \param{clique number} of $G$ is at most~5.

The partial order $\pi$ on the vertex set of $G$ is the reflexive transitive closure of the relation containing the following constraints:

\begin{enumerate}[(P1)]
    \item $s \prec v$ for every $v \in V(G) \setminus \{s\}$,\label{ui:p1}
    \item $t \prec v$ for every vertex $v$ in the backbone $B$,\label{ui:p2}
    \item $t \prec s' \prec t'$,\label{ui:p3}
    \item $x_i \prec \ell_j^k$ if the $k$-th literal in $c_j$ is $x_i$,\label{ui:p4}
    \item $\overline{x}_i \prec \ell_j^k$ if the $k$-th literal in $c_j$ is $\overline{x}_i$.\label{ui:p5}
\end{enumerate}

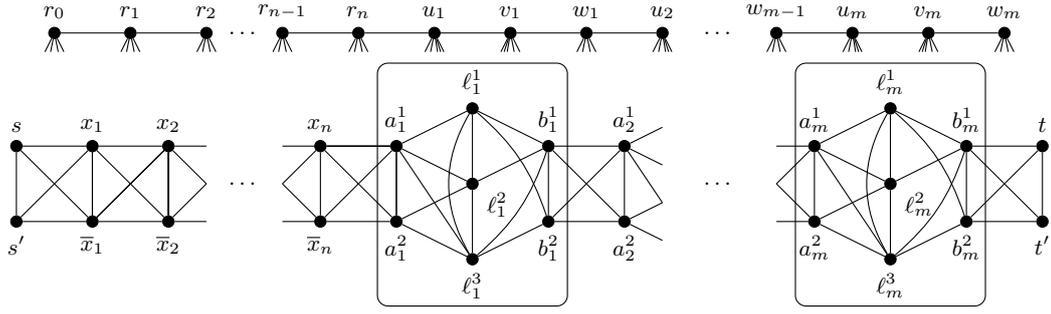
\begin{figure}
    \centering
        \begin{tikzpicture}
        \footnotesize

        \begin{scope}
        \node[vertex, label=90:$s$] (s) at (0,1) {};
        \node[vertex, label=-90:$s'$] (s') at (0,0) {};
        \node[vertex, label=90:$x_1$] (x1) at (1,1) {};
        \node[vertex, label=-90:$\overline{x}_1$] (nx1) at (1,0) {};
        \node[vertex, label=90:$x_2$] (x2) at (2,1) {};
        \node[vertex, label=-90:$\overline{x}_2$] (nx2) at (2,0) {};
        \node[vertex, label=90:$x_n$] (xn) at (4,1) {};
        \node[vertex, label=-90:$\overline{x}_n$] (nxn) at (4,0) {};



        \draw (s) -- (x1) -- (nx1) -- (s);
        \draw (s') -- (x1) -- (nx1) -- (s') -- (s);
        \draw (x1) -- (x2) --
               (nx2) -- (nx1) -- (x2) -- (nx1) -- (x1) -- (nx2) -- (x2);
        \draw (x2) --+ (0.5, 0);
        \draw (nx2) --+ (0.5, 0);
        \draw (x2) --+ (0.5, -0.5);
        \draw (nx2) --+ (0.5, 0.5);
        
        \node at (3,0.5) {$\dots$};
        \node at (3,2.5) {$\dots$};

        \node at (9.25,0.5) {$\dots$};
        \node at (9.25,2.5) {$\dots$};

        \draw (xn) -- (nxn);
        \draw (xn) --+ (-0.5,0);
        \draw (nxn) --+ (-0.5, 0);
        \draw (xn) --+ (-0.5, -0.5);
        \draw (nxn) --+ (-0.5, 0.5);

        \node[vertex, label=90:$a_1^1$] (a11) at (5,1) {};
        \node[vertex, label=-90:$a_1^2$] (a12) at (5,0) {};

        \draw(xn) -- (a11) -- (a12) -- (nxn) -- (a11) -- (xn) -- (a12);

        \node[vertex, label=90:$r_0$] (r0) at (0.5,2.5) {};
        \draw (r0) --+ (0.05,-0.25);
        \draw (r0) --+ (0.125,-0.25);
        \draw (r0) --+ (-0.125,-0.25);
        \draw (r0) --+ (-0.05,-0.25);
        \node[vertex, label=90:$r_1$] (r1) at (1.5,2.5) {};
        \draw (r1) --+ (0.05,-0.25);
        \draw (r1) --+ (0.125,-0.25);
        \draw (r1) --+ (-0.125,-0.25);
        \draw (r1) --+ (-0.05,-0.25);
        \node[vertex, label=90:$r_2$] (r2) at (2.5,2.5) {};
        \draw (r2) --+ (0.05,-0.25);
        \draw (r2) --+ (0.125,-0.25);
        \draw (r2) --+ (-0.125,-0.25);
        \draw (r2) --+ (-0.05,-0.25);
        
        \draw (r0) -- (r1) -- (r2);
        
        \end{scope}

        \begin{scope}
        \node[vertex, label=90:{$\ell_1^1$}] (x11) at (6,1.5) {};
        \node[vertex, label=-4:$\ell_1^2$] (x12) at (6,0.5) {};
        \node[vertex, label=-90:$\ell_1^3$] (x13) at (6,-0.5) {};
        \node[vertex, label=90:$b_1^1$] (C11) at (7,1) {};
        \node[vertex, label=-90:$b_1^2$] (C12) at (7,0) {};

        \draw (a11) -- (a12);
        \draw (a12) -- (x12);
        \draw (a12) -- (x13);
        \draw (a11) -- (x11);
        \draw (a11) -- (x12);
        \draw (a11) -- (x13);
        
        \draw (x11) -- (x12) -- (x13);
        \draw[bend left] (x13) to (x11);

        \draw[rounded corners] (4.75,2.1) rectangle (7.25,-1.12);
        \draw[rounded corners] (10.25,2.1) rectangle (12.75,-1.12);

        \node[vertex, label=90:$a_2^1$] (a21) at (8,1) {};
        \node[vertex, label=-90:$a_2^2$] (a22) at (8,0) {};

        \draw (C11) -- (C12) -- (a21) -- (a22) -- (C11) -- (a21);
        \draw (C12) -- (a22);
        
        \draw (C11) -- (x11);
        \draw (C11) -- (x12);
        \draw[bend angle=15, bend left] (C11) to (x13);
        \draw[bend angle=15, bend right]  (C12) to (x11);
        \draw (C12) -- (x13);
        
        \draw (a21) --+ (0.5,0.25);
        \draw (a21) --+ (0.5,-0.25);
        \draw (a21) --+ (0.5,-0.75);
        \draw (a22) --+ (0.5,-0.25);
        \draw (a22) --+ (0.5,0.25);



        

        
        \begin{scope}[xshift=-2.5cm]
        \node[vertex, label=90:$a_{m}^1$] (am1) at (13,1) {};
        \node[vertex, label=-90:$a_{m}^2$] (am2) at (13,0) {};
        \node[vertex, label=90:$\ell_{m}^1$] (xm1) at (14,1.5) {};
        \node[vertex, label=-4:$\ell_{m}^2$] (xm2) at (14,0.5) {};
        \node[vertex, label=-90:{$\ell_{m}^3$}] (xm3) at (14,-0.5) {};
        \node[vertex, label=90:$b_m^1$] (cm1) at (15,1) {};
        \node[vertex, label=-90:$b_m^2$] (cm2) at (15,0) {};
        \node[vertex, label=90:$t$] (t) at (16,1) {};
        \node[vertex, label=-90:$t'$] (t') at (16,0) {};

        \node[vertex, label=90:$w_{m-1}$] (wm-1) at (12.5,2.5) {};
        \draw (wm-1) --+ (0.05,-0.25);
        \draw (wm-1) --+ (0.125,-0.25);
        \draw (wm-1) --+ (-0.125,-0.25);
        \draw (wm-1) --+ (-0.05,-0.25);
        \node[vertex, label=90:$u_m$] (um) at (13.5,2.5) {};
        \draw (um) --+ (-0.05,-0.25);
        \draw (um) --+ (-0.125,-0.25);
        \draw (um) --+ (0.125,-0.25);
        \draw (um) --+ (0.0825,-0.25);
        \draw (um) --+ (0.04167,-0.25);
        \node[vertex, label=90:$v_m$] (vm) at (14.5,2.5) {};
        \draw (vm) --+ (0.05,-0.25);
        \draw (vm) --+ (0.125,-0.25);
        \draw (vm) --+ (-0.125,-0.25);
        \draw (vm) --+ (-0.0825,-0.25);
        \draw (vm) --+ (-0.04167,-0.25);
        \node[vertex, label=90:$w_m$] (wm) at (15.5,2.5) {};
        \draw (wm) --+ (0.05,-0.25);
        \draw (wm) --+ (0.125,-0.25);
        \draw (wm) --+ (-0.125,-0.25);
        \draw (wm) --+ (-0.05,-0.25);

        \draw (wm-1) -- (um) -- (vm) -- (wm);
        \end{scope}


        \draw (am1) --+ (-0.5,0.0);
        \draw (am1) --+ (-0.5,-0.5);
        \draw (am2) --+ (-0.5,-0.0);
        \draw (am2) --+ (-0.5,0.5);
        
        \draw (am1) -- (xm1);
        \draw (am1) -- (xm2);
        \draw (am1) -- (xm3);
        \draw (am2) -- (xm2);
        \draw (am2) -- (xm3);
        \draw (am2) -- (am1);

        \draw (xm1) -- (xm2) -- (xm3);
        \draw[bend left] (xm3) to (xm1);

        \draw (cm1) -- (xm1);
        \draw (cm1) -- (xm2);
        \draw[bend angle=15, bend left] (cm1) to (xm3);
        \draw[bend angle=15, bend right] (cm2) to (xm1);
        \draw (cm2) -- (xm3);

        \draw (cm2) -- (cm1) -- (t) -- (cm2) -- (t') -- (t);
        \draw (cm1) -- (t');

        \node[vertex, label=90:$r_{n-1}$] (rn-1) at (3.5,2.5) {};
        \draw (rn-1) --+ (0.05,-0.25);
        \draw (rn-1) --+ (0.125,-0.25);
        \draw (rn-1) --+ (-0.125,-0.25);
        \draw (rn-1) --+ (-0.05,-0.25);
        \node[vertex, label=90:$r_n$] (rn) at (4.5,2.5) {};
        \draw (rn) --+ (0.05,-0.25);
        \draw (rn) --+ (0.125,-0.25);
        \draw (rn) --+ (-0.125,-0.25);
        \draw (rn) --+ (-0.05,-0.25);
        \node[vertex, label=90:$u_1$] (u1) at (5.5,2.5) {};
        \draw (u1) --+ (-0.05,-0.25);
        \draw (u1) --+ (-0.125,-0.25);
        \draw (u1) --+ (0.125,-0.25);
        \draw (u1) --+ (0.0825,-0.25);
        \draw (u1) --+ (0.04167,-0.25);
        \node[vertex, label=90:$v_1$] (v1) at (6.5,2.5) {};
        \draw (v1) --+ (0.05,-0.25);
        \draw (v1) --+ (0.125,-0.25);
        \draw (v1) --+ (-0.125,-0.25);
        \draw (v1) --+ (-0.0825,-0.25);
        \draw (v1) --+ (-0.04167,-0.25);
        \node[vertex, label=90:$w_1$] (w1) at (7.5,2.5) {};
        \draw (w1) --+ (-0.05,-0.25);
        \draw (w1) --+ (-0.125,-0.25);
        \draw (w1) --+ (0.125,-0.25);
        \draw (w1) --+ (0.05,-0.25);
        \node[vertex, label=90:$u_2$] (u2) at (8.5,2.5) {};
        \draw (u2) --+ (-0.05,-0.25);
        \draw (u2) --+ (-0.125,-0.25);
        \draw (u2) --+ (0.125,-0.25);
        \draw (u2) --+ (0.0825,-0.25);
        \draw (u2) --+ (0.04167,-0.25);
        
        \draw (rn-1) -- (rn) -- (u1) -- (v1) -- (w1) -- (u2);
        \end{scope}
    \end{tikzpicture}
    \caption{Complete construction of the proof of \cref{thm:np-interval}. The boxes mark the clause gadgets.}
    \label{fig:unit-interval}
\end{figure}

\subparagraph*{From a $\pi$-extending path to a satisfying assignment} Let $\cP$ be a $\pi$-extending Hamiltonian path of $G$. Due to \pef{ui:p1}, the path $\cP$ starts with $s$. Let $\cP_{s,t}$ be the subpath of $\cP$ starting in $s$ and ending in $t$. 

\begin{claim*}
    For any $i \in \{1,\dots,n\}$, $\cP_{s,t}$ contains exactly one of the two vertices $x_i$ and $\overline{x}_i$.  
\end{claim*}

\begin{claimproof}
    Due to \pef{ui:p2}, the path $\cP_{s,t}$ cannot contain any of the vertices in the backbone. Any path between $s$ and $t$ that does not contain those vertices has to contain at least one of the variable vertices $x_i$ and $\overline{x}_i$ for every $i \in \{1,\dots,n\}$.

    Let $\cP_{t,s'}$ be the subpath of $\cP$ between $t$ and $s'$ and $\cP_{s',t'}$ be the subpath of $\cP$ between $s'$ and $t'$. Due to \pef{ui:p1} and \pef{ui:p3}, vertex $s$ is to the left of $t$ in $\cP$, $t$ is to the left of $s'$ in $\cP$ and $s'$ is to the left of $t'$ in $\cP$. Therefore, $\cP_{s,t}$, $\cP_{t,s'}$, and $\cP_{s',t'}$ do not share an inner vertex. All three paths have to use one of the vertices $x_i$, $\overline{x}_i$ and $r_i$. Hence, $\cP_{s,t}$ cannot contain both $x_i$ and $\overline{x}_i$.
\end{claimproof}

Due to this claim, we can define an assignment $\A$ of $\Phi$ in the following way: A variable $x_i$ is set to true in $\A$ if and only if the vertex $x_i$ is contained in $\cP_{s,t}$. Since $\cP_{s,t}$ cannot contain any of the $u_i$ or $v_i$, it is clear that $\cP_{s,t}$ must contain at least one of the vertices $\ell_i^1$, $\ell_i^2$, and $\ell_i^3$ for any $i \in \{1,\dots,m\}$. As those vertices can only be visited by $\cP$ if their respective variable vertex is to the left of them in $\cP$, it follows that $\A$ is an fulfilling assignment of $\Phi$.

\subparagraph*{From a satisfying assignment to a $\pi$-extending Hamiltonian path} Let $\A$ be a satisfying assignment of $\Phi$. We define the ordered Hamiltonian path $\cP$ as follows. We start in $s$. Then we successively visit the variable vertices in such way that we visit $x_i$ if the variable $x_i$ is set to true in $\A$ and otherwise we visit $\overline{x}_i$.  Afterwards, we visit $a_1^1$. As $\A$ is a fulfilling assignment, there is at least one of the vertices $\ell_1^1$, $\ell_1^2$, and $\ell_1^3$ that can be visited as the next vertex. We choose one of these vertices as next vertex and then visit $b_1^1$. We repeat this process for each $i \in \{2, \dots, m\}$. Afterwards, we visit vertex $t$. As the backbone vertices were only restricted by \pef{ui:p2}, we now can use the complete backbone in decreasing order to reach vertex $s'$. Next we visit all the remaining variable vertices in increasing order.
In the clause gadget for \(c_1\), we first visit $a_1^2$. 
Then, we visit all previously unvisited literal vertices in an order that visits \(\ell_1^2\) or \(\ell_1 ^3\) first and \(\ell_1^1\) or \(\ell_1^3\) last, followed by \(b_1^2\).
We repeat this procedure for all $i \in \{2, \dots, m\}$ and finally visit $t'$. The resulting path is a $\pi$-extending Hamiltonian path of $G$.
\end{proof}

The \param{bandwidth} of a proper interval graph is equal to its clique number minus~1~\cite[Theorem~4.1]{kaplan1996pathwidth}. Hence, the following holds.

\begin{theorem}\label{cor:NPhard_pathwidth}
     POHPP is \NP-complete on graphs of \param{bandwidth}, \param{pathwidth} or \param{treewidth}~4.
    \end{theorem}

We now build on the ideas presented in \cref{thm:np-interval} to show that  POHPP is \NP-complete on grid graphs of height at least 7.

\begin{theorem*}\label{thm:np-grid}
   POHPP is \NP-complete on grid graphs of height $7$.
\end{theorem*}

\begin{proof}
    We present a reduction from 3-SAT to  POHPP that shares the main ideas of the proof of \cref{thm:np-interval}.
    Let \(\Phi\) be a 3-SAT formula over variables \(x_1,\dots, x_n\) with clauses \(c_1,\dots, c_m\).
    We construct an instance \((G,\pi)\) for  POHPP where \(G\) is an \(7 \times (5n + 4m + 6)\) grid graph.
    In order to be able to argue about different parts of \(G\), it is conceptually subdivided into different gadgets. We also name some special vertices within these gadgets.
    For a gadget \(Z\), we use \(Z[a,b]\) to denote the vertex in the \(a\)th row and \(b\)th column of the gadget. The colors in parathesis for each gadget refer to the colors used in \cref{fig:grid_gadgets}.
    \begin{description}
        \item[Start gadget \(S\) (rose):] A \(7\times 3\) grid graph. We denote the vertex \(S[3,3]\) by \(s\).
        \item[Variable switch gadgets \(Y_i\) (blue):] A \(7\times 1\) grid graph for \(i=1,\dots n\).
        \(Y_i\) is assigned to variable \(x_i\). 
        \item[Variable gadget \(X_i\) (yellow):] A \(7\times 4\) grid graph for \(i=1,\dots, n\).
        \(X_i\) is assigned to variable \(x_i\).
        We call the vertices \(X_i[3,2]\) and \(X_i[3,3]\) the \emph{negative variable vertices} of \(X_i\). The vertices \(X_i[5,2]\) and \(X_i[5,3]\) are the \emph{positive variable vertices}. 
        \item[Middle gadget \(M\) (red):]  A \(7\times 2\) grid graph.
        \item[Clause switch gadget \(D_j\) (light blue):]  A \(7\times 2\) grid graph \(D_j\) assigned to \(c_j\) for \(j=1,\dots, m\).
        \item[Clause gadget \(C_j\) (green):]  A \(7\times 2\) grid graph  assigned to \(c_j\) for \(j=1,\dots m\).
        The vertices \(C_j[a+2,1]\) and  \(C_j[a+2,2]\) are assigned to the literal \(\ell_j^a\).
        \item[End gadget \(T\) (rose):] A \(7\times 1\) grid graph. We denote the vertex \(T[5,1]\) by \(t\).
    \end{description}
    \(G\) is made up of the gadgets in the following order: \(S, Y_1,X_1,\dots, Y_n, X_n, M, D_1,C_1,\dots, D_m, \allowbreak C_m, T\), see \cref{fig:grid_gadgets}.
    The partial order \(\pi\) is the reflexive, transitive closure of the  following constraints:
    \begin{enumerate}[(P1)]
   \item $s\prec v$ for all \(v\in V(G) \setminus \{s\}\). \label{item:grid_start}
        \item $t\prec v$ for all \(v\in S\setminus\{s\}\).
        \item $t\prec v$ for all vertices \(v\) in Rows~\(1,2,6\) and \(7\).
        \item $t\prec v$ if \(v\) is in Row~\(4\) of some \(X_i\)
        \item $u\prec v$ if \(u\) is a negative variable vertex in \(X_i\) and \(v\) in \(C_j\) is assigned to a literal \(\overline{x}_i\). \label{item:grid_lit1}
        \item $u\prec v$ if \(u\) is a positive variable vertex in \(X_i\) and \(v\) in \(C_j\) is assigned to a literal \(x_i\).\label{item:grid_lit2}
    \end{enumerate}
    We call \pef{item:grid_lit1} and \pef{item:grid_lit2} \emph{literal constraints}.
    \begin{figure}
        \centering
        \includegraphics[width=\textwidth]{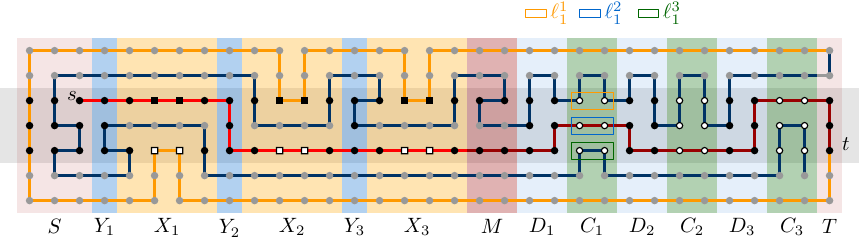}
        \caption{The gadgets for \cref{thm:np-grid} combined to represent a formula. The gray vertices come after \(t\) in the partial order \(\pi\).
        A black square is negative variable vertex, a white square is a positive variable vertex.
        A white disk marks a literal vertex.
        }
        \label{fig:grid_gadgets}
    \end{figure}

Now we show that the instance described above has a $\pi$-extending Hamiltonian path if and only if, the given formula \(\Phi\) is satisfiable.

\subparagraph*{From a $\pi$-extending Hamiltonian path to a satisfying assignment}

    Let \(\mathcal{P}\) be a \(\pi\)-extending Hamiltonian path of \(G\).
     By \pef{item:grid_start}, \(\cP\) starts in \(s\).
    Let \(\mathcal{P}_{s,t}\) be the prefix of \(\mathcal{P}\) ending in \(t\).
    Observe that by the definition of \(\pi\), all vertices that are in Rows \(1,2,6\) or \(7\) or in Row~4 of some variable gadget cannot lie on \(\mathcal{P}_{s,t}\).

    \begin{claim*}\label{claim:pos_neg}
    For any $i \in \{1,\dots,n\}$, $\cP_{s,t}$ either contains both positive or both negative variable vertices of \(X_i\).
\end{claim*}
\begin{claimproof}
Assume for a contradiction that there is an \(i\), such that one of \ \(X_i[3,2]\) or \(X_i[3,3]\) and one of \(X_i[5,2]\) or \(X_i[5,3]\) lie on \(\cP_{s,t}\).
As rows \(1,2,4,6,7\) of \(X_i\) are all after \(t\) in any linear extension of \(\pi\), a valid prefix path \(\cP_{s,t}\) cannot switch rows within \(X_i\).
In particular, it can only cross  \(X_i\) in one row.
\(\cP_{s,t}\) thus first visits either the positive or the negative variable vertices and then exits \(X_i\) on the right and then reenters it from the right, see \cref{fig:grid-assignment} for an illustration.
After traversing \(X_i\) twice, all vertices of \(X_i\) that can be visited before \(t\) are already on the path.
Thus, there are no possible vertices left to cross \(X_i\) again and the path cannot reach \(t\), a contradiction.
\end{claimproof}
\begin{figure}
    \centering
      \includegraphics[page=2]{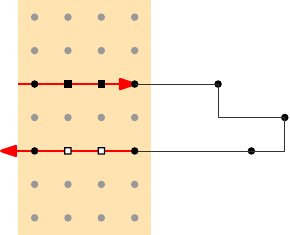}
    \caption{The path \(\cP_{s,t}\) is stuck when traversing both row \(3\) and row \(5\) of \(X_i\)}
    \label{fig:grid-assignment}
\end{figure}

\cref{claim:pos_neg} allows us to define an assignment of the variables of \(\Phi\) as follows.
Set \(x_i\) to false if and only if \(\cP_{s,t}\) traverses \(X_i\) through the negative variable vertices.
Consider a clause gadget \(C_j\) and a vertex \(v\) in \(C_j\) that is visited by \(\cP_{s,t}\).
As only vertices in  rows \(3,4\) or \(5\) are on \(\cP_{s,t}\), the vertex \(v\) has an assigned literal.
Let \(\ell_j^a\) be the literal assigned to \(v\).
If \(\ell_j^a = x_i\) for some \(i\), then by \pef{item:grid_lit2}, the positive variable vertices of \(x_i\) are before \(v\) on \(\cP_{s,t}\). 
Thus the assignment specified above sets \(x_i=1\), satisfying \(c_j\).
If \(\ell_j^a = \overline{x_i}\), by \pef{item:grid_lit1}, the negative variable vertices of \(x_i\) are visited before \(v\) on \(\cP_{s,t}\). 
As in this case \(x_i=0\), the literal \(\ell_j^a\) satisfies \(c_j\).

\begin{figure}
    \centering
    \includegraphics[page=4]{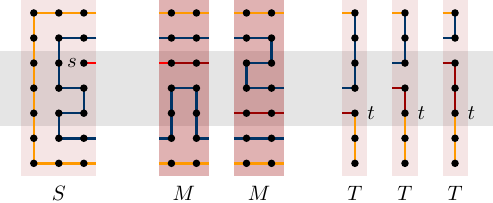}
    \caption{The path \(\cP\) in the start, middle  and end gadget}
    \label{fig:grid_start}
\end{figure}
\begin{figure}
    \centering
    \includegraphics[page=5]{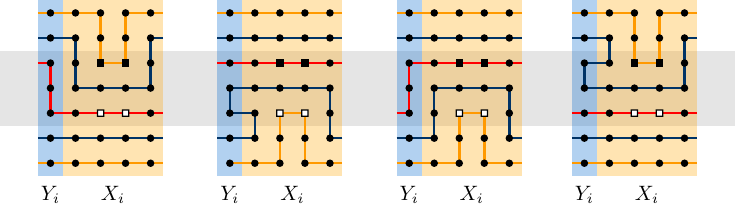}
    \caption{The path \(\cP\) in the variable and variable switch gadgets}
    \label{fig:grid_variable}
\end{figure}
\begin{figure}
    \centering
    \includegraphics[page=6, width=\textwidth]{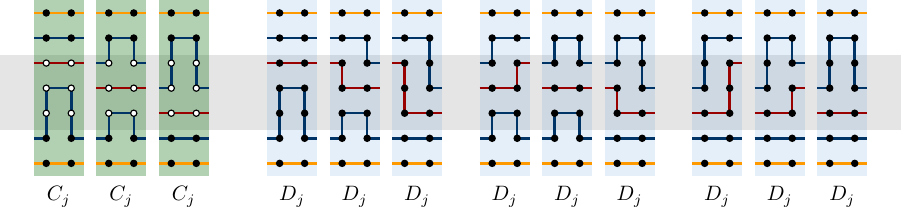}
    \caption{The path \(\cP\) in the  clause and  clause switch gadgets}
    \label{fig:grid_clause}
\end{figure}

\subparagraph*{From a satisfying assignment to a $\pi$-extending Hamiltonian path}
Assume that \(\Phi\) has a satisfying assignment. 
Then we show that there is a Hamiltonian path \(\cP\) that is valid for \(\pi\).
We describe the path in three parts.
The first part \(\cP_{s,t}\) is a valid prefix that connects \(s\) and \(t\).
This part is marked in red in \cref{fig:grid_gadgets,fig:grid_start,fig:grid_variable,fig:grid_clause}.
The second part (\(\cP_2\)) of the path visits the remaining variable vertices.
This part is drawn in orange in the figures.
In the last part (\(\cP_3)\), all remaining vertices, in particular those in the clause gadgets are visited (see the blue path in the figures).

The path \(\cP_{s,t}\) starts in \(s\) and then enters \(Y_1\).
If \(x_1 = 0\), then it stays in row~\(3\) and crosses \(X_1\) through the negative variable vertices.
In the other case, it visits the vertices $Y_1[3,1]$, $Y_1[4,1]$, and $Y_1[5,1]$  and then crosses \(X_1\) through the positive variable vertices.
The same process is repeated for all remaining variables.
If \(x_{i-1}\) and \(x_{i}\) have the same assignment, then \(Y_i\) is simply crossed in row~\(3\) or \(5\).
In the other case, the three vertices in Row 3, 4, and 5 of \(Y_i\) are added to \(\cP_{s,t}\).
This is continued until \(M\) is reached. \(M\) is crossed without changing the row.
Let \(\ell_j^a\) be an arbitrary literal that satisfies the clause \(c_j\), then \(\cP_{s,t}\) crosses the clause gadget \(C_j\) in the vertices assigned to \(\ell_j^a\).
If this row changes between consecutive clauses, \(D_{j}\) is used to get the path to the correct row. 
If the path is moving up, this is done in the second column of \(D_j\). If the path is moving down, this is done in the first column. 
In \(T\) the path moves directly down to \(t\).
To see that \(\cP_{s,t}\) is the prefix of a \(\pi\)-extending Hamiltonian path, we first note that only vertices \(v\) with \(v\prec t\) are visited.
Furthermore, if a vertex in a clause gadget is visited that is assigned to a negative literal, then the corresponding negative variable vertices were visited before. 
Analogously, all positive variable vertices are visited before vertices assigned to a positive literal.
This makes sure that the literal constraints involving the visited literal vertices are satisfied.

For the subpath \(\cP_2\), all vertices in variable gadgets that are not on \(\cP_{s,t}\) are visited.
First, \(\cP_2\) visits \(T[6,1]\) and \(T[7,1]\) (see \cref{fig:grid_start} on the right side).
Then it crosses all clause and switch clause gadgets as well as the middle gadget in row~\(7\), until it reaches \(X_n[7,4]\).
The variable gadgets are now visited in the following way:
For each variable, if \(x_i=0\) in the satisfying assignment, then \(\cP_2\) visits the positive variable vertices as depicted by the orange path in the middle two pictures of \cref{fig:grid_variable}.
In the other case, it stays in Row~\(7\), see the leftmost and rightmost figure in \cref{fig:grid_variable}.
Afterwards, \(\cP_2\) continues through \(S\), as indicated by the orange curve in the left figure  \cref{fig:grid_start}.
Then it follows a similar pattern in the top to reach \(T[1,1]\), visiting the negative variable vertices of variables \(x_i\) with \(x_i=1\) as shown by the orange curves in the left and right picture of \cref{fig:grid_variable}.
After leaving \(X_n[1,4]\) all remaining vertices in Row~\(1\) are visited until \(\cP_2\) reaches \(T[1,1]\).

In the final phase, all remaining vertices are visited. Recall that this part of the path is called \(\cP_3\) and it is drawn in blue in all figures.
For the clause and clause switch gadgets, \(\cP_3\) first visits all vertices above \(\cP_{s,t}\) maintaining the following invariant:
If \(\cP_{s,t}\) crosses from one gadget to the other in Row~\(k\), then \(\cP_3\) crosses from one gadget to the other in Row~\(k-1\).
In  \cref{fig:grid_clause} all possible realizations of \(\cP_{s,t}\) within the clause and clause switch gadgets are shown, together with a path \(\cP_3\).
For variable and variable switch gadgets, the invariant for the path is that it enters \(X_i\) and leaves \(Y_i\) in Row~\(2\). 
The realization of \(\cP_3\) for all possible assignments can be seen in \cref{fig:grid_variable}.
The middle gadget \(M\) is used to make the change between these two invariants. In the two middle parts of  \cref{fig:grid_start} one can see the two possible realizations of \(\cP\) in \(M\).
The blue path in the left picture in \cref{fig:grid_start} shows how the path continues in \(S\).

Finally the remaining vertices in \(S\) and below \(\cP_{s,t}\) are visited in a similar way but with a slightly different invariant. 
Each block \((Y_i,X_i)\) for a variable, as well as all remaining gadgets are entered and left in row~\(6\).
The path ends in \(C_m[6,2]\).
See \cref{fig:grid_variable} for possible realizations maintaining this invariant for all possible assignments of the variables.

By maintaining the invariants mentioned above, it is clear, that the path \(\cP\) that results from combining \(\cP_{s,t}, \cP_2\) and \(\cP_3\) is indeed a Hamiltonian path.
It is left to argue that the path \(\cP\) is indeed \(\pi\)-extending.
We have already argued that \(\cP_{s,t}\) is a prefix of a \(\pi\)-extending Hamiltonian path.
Therefore, the first four types of constraints of $\pi$ as well as the literal constraints for all variable vertices on \(\cP_{s,t}\) are fulfilled.
To see that the remaining literal constraints are satisfied, we argue as follows. If some vertex assigned to a literal is visited after \(\cP_{s,t}\), then all remaining positive and negative variable vertices have already been visited in $\cP_2$.
\end{proof}

\begin{theorem}
    POHPP is \NP-complete on grid graphs of height \(h\) for  \(h \geq 7\).
\end{theorem}
\begin{proof}
The same construction as in the proof of \Cref{thm:np-grid} can be used with a slight modification.
All gadgets are now \(h \times k\) grids.
Except for the variable gadgets, that will be a \(h \times 5\) grid, all gadget widths are the same as in the construction of \cref{thm:np-grid}.
The vertices in Rows~\(8,\dots, h\) take over the ordering constraints of the vertices in Row~\(7\) in the same column.
The proof then follows by the same arguments as that of \cref{thm:np-grid}, replacing the part of \(\cP_2\)  that visits the vertices in Row~\(7\) by a space filling curve that enters and leaves each variable block and each remaining gadget in Row~\(h\).
See \cref{fig:grid_larger_seven} for an illustration. \qedhere
\begin{figure}
    \centering
    \includegraphics[page=8, width=\textwidth]{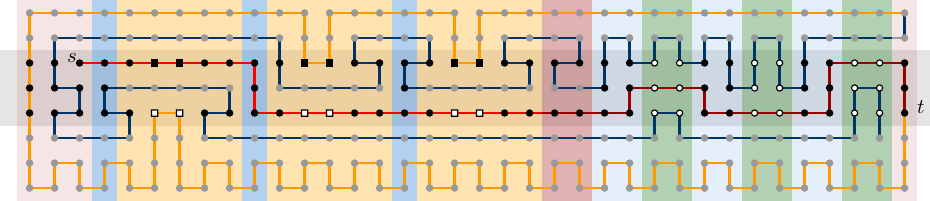}
    \caption{Example for the construction with height \(8\)}
    \label{fig:grid_larger_seven}
\end{figure}
\end{proof}

\subsection{Algorithms}
In the section above, we showed that POHPP is \NP{}-complete for graphs with \param{treewidth} and \param{pathwidth} at least four. 
In this section we give polynomial time algorithms for \param{pathwidth} at most \(3\) and \param{treewidth} at most \(2\).

\subsubsection{Pathwidth at most 3}\label{subsec:pathwidth3}

Let \(X_1,\dots, X_k\) be a path decomposition of width \(3\) of a graph \(G\).
For our algorithm we assume that every bag contains exactly four vertices. Furthermore, we assume that \(X_i\) and \(X_{i+1}\) differ in exactly two vertices.
We call the unique vertex \(u\in X_{i}\setminus X_{i+1}\) the vertex that is \emph{forgotten} in \(X_{i+1}\) and the unique vertex \(w\in X_{i+1}\setminus X_i\) the vertex that is \emph{introduced} in \(X_{i+1}\).
Furthermore, let \(G_i\) be the induced subgraph on \(V_i = \bigcup_{j=1}^i X_j\).
A general path-decomposition of width~3 can be computed in linear time~\cite{bodlaender1996linear-time,furer2016faster}.
It might not fulfill the the constraint that \(|X_i|= 4\) and that consecutive bags differ in exactly two vertices. It can however be brought into this form in linear time.

As each vertex set \(X_i\) is a separator of \(G\), the following observation holds:
 \begin{observation*}\label{obs:forgotten_connected}
 The vertex \(u\) forgotten in \(X_{i+1}\) has no edge to a vertex in \(V\setminus V_i\).
 The vertex \(w\) introduced in \(X_{i+1}\) has no edge to a vertex in \(V_i  \setminus X_{i+1}\).
 \end{observation*}

 Our algorithm is based on the folklore dynamic programming algorithm that solves \textsc{Hamiltonian Path} for graphs of bounded \param{pathwidth}.
 We will first sketch that algorithm and then describe the modifications needed when considering partial order constraints.

Let \(\cP\) be an ordered Hamiltonian path in \(G\) and \(\cP_i = \cP \cap G_i\) be the set of paths induced by \(V_i\). 
For \(P=(v_1,\dots, v_\ell)\in \cP_i\) we call \(v_1\) the \emph{start vertex}, \(v_\ell\) the \emph{end vertex} and the remaining vertices the \emph{interior vertices}. 
In some cases, we refer to the start and end vertices as \emph{terminal vertices}.
If \(v_1\) is the start vertex of \(\cP\) or \(v_\ell\) is the end vertex of \(\cP\), we call \(v_1\) or \(v_\ell\) the \emph{global start vertex} or \emph{global end vertex}, respectively.

If \(|P|=1\), we call \(P\) a \emph{trivial path}.
 A path \(\cP\) induces a \emph{signature} that encodes the interaction of \(\cP_i\) with the vertices in \(X_i\).
 The subproblems in the dynamic program are then all pairs of a bag \(X_i\) and a signature \(\sigma\) for that bag and the entry for that pair is \texttt{true} if there is a set of path that visit all vertices in \(V_i\) and interact with \(X_i\) as indicated by the signature and \texttt{false} otherwise.
 The recurrence in the dynamic program then considers all entries for \(X_{i-1}\) with a signature \(\gamma\) that is \emph{compatible} to $\sigma$, i.e., both $\gamma$ and $\sigma$ can hypothetically be induced by the same Hamiltonian path.
 If there is one such signature $\gamma$ with value \texttt{true} and the vertex introduced in \(X_i\) can be connected to the vertices in \(X_i \cap X_{i-1}\) to form \(\sigma\), then the entry is \texttt{true}.

 The number of signatures and pairs of compatible signatures is bounded by a function depending on the width of the bags and, thus, the algorithm is an \FPT{} algorithm when parameterized by \param{pathwidth}.

When not considering ordering constraints, the information if there is a set of paths that visit every vertex in \(V_i\) exactly once and that interface with \(V\setminus V_i\) as indicated by \(\sigma\) is enough.
However, in POHPP the order in which the vertices are visited is relevant. 
 Thus, we also need information about the partition of the vertices in \(X_i\) to the paths.
 One could naively extend the dynamic program above to store all possible partitions of the vertices in \(V_i\) to the paths. 
 However, as the number of these partitions grows exponentially with~\(n\) this approach is not feasible in general.
For \param{pathwidth} \(3\), however, we can show that the relevant information can be maintained and computed efficiently.

The following lemmas are the basis for our algorithm.
\begin{lemma*}\label{lem:pw3_storinternal}
    Let \(\cP\) be a Hamiltonian path, let \(1\leq i < k\) and \(u\) be the vertex forgotten in \(X_{i+1}\). Then
    \begin{enumerate}
        \item \(u\) is either the global start vertex, the global end vertex, or an internal vertex of a path in \(\cP_i\), and
        \item \(\cP_i\) contains at most two isolated vertices.
    \end{enumerate}
\end{lemma*}
\begin{proof}
First note, that \(u\) cannot be an isolated vertex as, by \cref{obs:forgotten_connected}, it cannot be connected to \(\cP\) later.
Now, if \(u\) is the terminal vertex of a path in \(\cP_i\), then, also by \cref{obs:forgotten_connected}, it is the global start or end vertex. In the other case, it is an internal vertex and the first statement follows.

For the second part, it is clear that there cannot be four isolated vertices, as \(u\in X_i\) is not an isolated vertex by the argument above.
Furthermore, if there are exactly three isolated vertices, then \(u\) is the terminal vertex of a non-trivial path containing elements of $V_i \setminus X_i$ and, thus, $u$ is not the global start or end vertex. However, by the first part of this observation, if \(u\) is not the global start or end vertex, it is an interior vertex, a contradiction.
\end{proof}

\begin{lemma}\label{lem:pw3_form}
     For \(1\leq i < k\), \(\cP_i\) has one of the following forms:
    \begin{enumerate}
        
    \item \(\cP_i\) contains exactly one non-trivial path together with at most two isolated vertices.
    \item \(\cP_i\) contains a non-trivial prefix and a non-trivial suffix of $\cP$.
    \item \(\cP_i\) contains a non-trivial midpart and a non-trivial prefix or a non-trivial suffix of $\cP$.
\end{enumerate}
Furthermore, \(\cP_i = \{\cP\}\) if and only if \(i=k\).
\end{lemma}
\begin{proof}
First note that \(\cP_i\) cannot contain two prefixes or two suffixes as the start vertex of a prefix and the end vertex of a suffix are the unique terminal vertices of \(\cP\).
If \(\cP_i\) contains only one path, then there are at most two isolated vertices, due to \cref{lem:pw3_storinternal}.

Now assume that \(\cP_i\) contains two midparts or more than two non-trivial paths. Then, by a simple counting argument there is no internal vertex of a path \(P\in \cP_i\) in \(X_i\), a contradiction to \cref{lem:pw3_storinternal} and, thus, the statement holds.

If \(\cP_i = \{\cP\}\), then \(V_i=V\) has to hold. This is only the case for \(V_k\) and thus the last part of the lemma holds.
\end{proof}
To ease some of the arguments, we will consider \(\cP\) to be both a prefix and a suffix.
The following lemma helps us to reduce the number of relevant partitions of the vertices to the paths.
\begin{lemma}\label{obs:change}
    For each path \(\cP\) such that \(\cP_{i}\) contains a non-trivial midpart and a non-trivial prefix (or suffix), there is an index \(\ell \leq i\) such that all \(\cP_j\) for \(j=\ell,\dots, i\) contain a non-trivial midpart and a non-trivial prefix (or suffix) and either \(\ell=1\) or \(\cP_{\ell-1}\) contains only one non-trivial path.
\end{lemma}
\begin{proof}
We show the statement for the case of a prefix.
It is clear that there is an \(\ell\) such that \(\ell=1\) or \(\cP_{\ell-1}\) does not contain both a non-trivial midpart and a non-trivial prefix.
Now assume \(\ell \neq 1\).
    First observe that if \(\cP_{\ell-1}\) contains a suffix, then all \(\cP_{j}\) with \(j\geq \ell-1\) will contain a suffix.
    Thus, if \(\cP_\ell\) contains a midpart and a prefix, then \(\cP_{\ell-1}\) cannot contain a suffix and the statement follows from \cref{lem:pw3_form}.
\end{proof}

We define the \emph{signature} \(\sigma\) of a bag \(X_i\) as a mapping of the vertices in \(X_i\) to the possible types of vertices and paths. 
If \(\sigma\) induces a midpart and another path, then \(\sigma\) also contains a value \(\ell\) for \(1\leq \ell \leq i\). If \(\ell\neq i\), then we also store a signature \(\tau\) for \(X_{\ell}\).
Let \(\cP(\sigma)\) be the set of paths induced by \(\sigma\).
Furthermore, let \(\sigma_t(v)\in \{\texttt{start},\texttt{end},\texttt{int}\}\) be the type assigned to \(v\) by \(\sigma\) and \(\sigma_p(v) \in \{\texttt{pre},\texttt{suf},\texttt{mid}\}\) be the type of path assigned to \(v\) by \(\sigma\).

We call a signature \emph{valid} if \(\cP(\sigma)\) and \(\sigma_t\) do not violate \cref{lem:pw3_storinternal} and \cref{lem:pw3_form}.
Intuitively, a signature \(\sigma\) encodes the following information for a path that is a candidate for a \(\pi\)-extending Hamiltonian path: How does the path interact with \(X_i\)? If there is a non-trivial mipart and another non-trivial path in $X_i$, then in which bag \(X_\ell\) did the second path appear? And, finally, how does the path interact with \(X_{\ell}\)?
In the following we will write \(\ell(\sigma)\) and \(\tau(\sigma)\) for the values \(\ell, \tau\) stored with a signature \(\sigma\). 

 We say that the signatures \(\gamma\) for \(X_{i-1}\) and \(\sigma\) for \(X_{i}\) are \emph{compatible} if they are both valid and there is a way to extend the paths induced by \(\gamma\) with the vertex \(w\) introduced in \(X_{i}\) to get the signature \(\sigma\).
 If \(\sigma\) induces a midpart and another path and \(\ell(\sigma) < i\), then \(\sigma\) and \(\gamma\) are only compatible if \(\ell(\sigma) = \ell(\gamma)\) and \(\tau(\sigma) = \tau(\gamma)\) for \(\ell(\sigma) \leq i-2\) and \(\tau(\sigma) = \gamma = \tau(\gamma)\) if \(\ell(\sigma) = i-1\).
Furthermore, if \(\ell(\sigma)=i\), then a signature \(\gamma\) is only compatible to \(\sigma\) if \(\gamma\) induces only one non-trivial path.
\begin{lemma*}\label{lem:pw3_onlyone}
Let \(\sigma\) be a signature for \(X_i\) such that \(\cP(\sigma)\) contains a non-trivial midpart and another non-trivial path. 
Then there is at most one signature \(\gamma\) for \(X_{i+1}\) that is compatible to \(\sigma\) such that \(\cP(\gamma)\) contains a non-trivial midpart and another non-trivial path.
\end{lemma*}
\begin{proof}
As a suffix or prefix cannot disappear between compatible signatures, the non-midparts in $\cP(\sigma)$ and $\cP(\gamma)$ must be of the same type.
Without loss of generality, assume that the non-midpart is a prefix.
Consider the bags \(X_{i}, X_{i+1}\) and \(X_{i+2}\).
Note that \(X_{i+2}\) exists as \(\cP(\gamma)\) contains a midpart and another path and such a signature is not valid for \(X_k\).

See \cref{fig:pw3_onlyone}  for an illustration of the following arguments. Let \(w\) be the vertex introduced in \(X_{i+1}\) and \(u\) be the vertex forgotten in \(X_{i+1}\).
Furthermore, let \(u'\) be the vertex forgotten in \(X_{i+2}\).
Then, by \cref{lem:pw3_storinternal}, \(u\) is either the global start vertex or \(\sigma_t(u) = \texttt{int}\).
As \(u\) is forgotten in \(X_{i}\) and \(\cP(\gamma)\) contains a prefix, \(u'\) cannot be the global start vertex and, thus, \(\gamma_t(u')=\texttt{int}\).
 All other vertices in \(X_{i+1}\) have to be mapped to terminal vertices to fulfill the assumption on \(\gamma\) and, thus, \(u'\) is the unique vertex that is assigned to be internal by \(\gamma\).
 Similarly, \(u\) is the only vertex, if any, that is assigned to be internal by \(\sigma\).

On the other hand, note that \(\gamma_t(w)\neq \texttt{int}\) since connecting \(w\) to two vertices in \(X_{i+1}\setminus X_{i}\) would either close a cycle, if they are in the same path, or it would connect the paths, a contradiction to the assumption on \(\gamma\).
Additionally, \(w\) cannot be the global start vertex since \(\cP(\sigma)\) already contains a prefix and, therefore, the global start vertex is in \(V_{i}\).
Thus, \(\gamma_t(w)\in \{\texttt{start}, \texttt{end}\}\) and it follows that \(w\neq u'\).
As \(\gamma_t(u') = \texttt{int}\)  but \(\sigma_t(u')\in \{\texttt{start}, \texttt{end}\}\), in any path that has signature \(\sigma\) in \(X_{i}\) and induces a midpart and a prefix in \(X_{i+1}\), \(w\) is connected to \(u'\).
Without loss of generality, assume \(\sigma_t(u') = \texttt{start}\). 
Then \(\gamma_t(w) = \texttt{start}\), \(\gamma_t(u')= \texttt{int}\) and for all other vertices in \(v\in X_i \cap X_{i+1} \cap X_{i+2}\) we have \(\gamma_t(v) = \sigma_t(v)\).
Furthermore, \(\gamma_p(v) = \sigma_p(v)\) for \(v\in X_{i+1}\cap X_i\) and  \(\gamma_p(w)=\gamma_p(u')\).
Thus all signatures \(\gamma\) that are compatible to \(\sigma\) have the same mappings \(\gamma_t\) and \(\gamma_p\).
As \(\ell(\sigma)  = \ell(\gamma)\) and \(\tau(\sigma) = \tau(\gamma)\) holds for every pair of compatible signatures, the statement follows.
 \qedhere
 \begin{figure}
     \centering
     \includegraphics{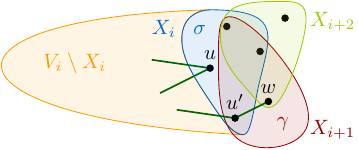}
     \caption{Illustration of the proof of \cref{lem:pw3_onlyone}. Vertex \(u'\) is a terminal vertex in the signature \(\sigma\) of $X_i$ but an internal vertex in the signature \(\gamma\) of $X_{i+1}$. Thus, it is connected to \(w\).}
     \label{fig:pw3_onlyone}
 \end{figure}
\end{proof}

Given a signature \(\sigma\) for \(X_i\), a \emph{path mapping} \(f_\sigma\) assigns each vertex \(v\in V_i\) to one of the paths induced by \(\sigma\).
In particular \(f_\sigma(v) = \sigma_p(v)\) for \(v\in X_i\), i.e., the vertices in \(X_i\) are mapped according to \(\sigma\).
\begin{observation*}\label{obs:pw:unique_partition}
    Let \(\sigma\) be a signature for \(X_i\).
    If \(i=1\) or \(\sigma\) induces only one non-trivial path or \(\ell(\sigma)=i\), then there is only one possible path mapping \(f_\sigma\).
\end{observation*}
\begin{proof}
    In the first two cases, the statement follows directly. For the third case, let \(w\) be the vertex introduced in \(X_i\). 
    As \(\ell(\sigma)=i\), a new path appears by adding \(w\). Thus \(w\) is connected to previously isolated vertices.
    This implies, that \(w\) together with any other vertices in \(X_i\) that are mapped to the same path are the only vertices in that path. All other vertices in \(V_i\) are mapped to the other path.
\end{proof}
If \(\sigma\) induces a midpart and another path, we additionally request that the vertices in \(V_{\ell(\sigma)}\) are mapped by \(f_\sigma\) according to the unique path mapping \(f_{\tau(\sigma)}\).
A path mapping \emph{contradicts} \(\pi\) if there is no possible order in which the vertices  in \(V\setminus V_i\) can be appended and prepended to the paths without violating the constraints given by \(\pi\).
If $v$ is part of the prefix, then the predecessors of $v$ in $\pi$ must not be contained in $V \setminus V_i$ or in a midpart. Equivalently, if $v$ is in the suffix, then the successors of $v$ in $\pi$ must not be contained in $V \setminus V_i$ or in a midpart if it exists. If $v$ is at the beginning (end) of a midpart, then the predecessors (successors) of $v$ in $\pi$ must not be part of that midpart.

Now we can define the dynamic program that will solve POHPP:
For  each bag \(X_i\) and each valid signature  \(\sigma\) for \(X_i\), define the subproblem \(D[i,\sigma]\).
If there is a path mapping \(f_\sigma\) for \(X_i\) that does not contradict \(\pi\) then \(D[i,\sigma]=f_\sigma\), otherwise \(D[i,\sigma]=\bot\).
The algorithm considers increasing values of \(i\).
In the base case \(D[1,\sigma]\), there is only one possible path partition \(f_\sigma\), due to \cref{obs:pw:unique_partition}.
If it contradicts \(\pi\), then set \(D[1,\sigma]=\bot\) and, otherwise, set \(D[1,\sigma]=f_\sigma\).

For an entry \(D[i,\sigma]\) with \(i\geq 2\), iterate over all valid signatures \(\gamma\) for \(X_{i-1}\) that are compatible to \(\sigma\) and where \(D[i-1,\gamma]\neq \bot\).
Let \(w\) be the vertex introduced in \(X_{i}\) and \(f_\sigma^\gamma\) the path partition that extends \(D[i-1,\gamma]\) by assigning \(w\) to the path defined by \(\sigma\).
As \(\sigma\) and \(\gamma\) are consistent, the position of \(w\) in the paths is uniquely determined. Iterate over the vertices \(V\setminus \{w\}\) to explicitly check if adding \(w\) in this unique position contradicts \(\pi\). 

If there is a \(\gamma\) such that \(w\) can be added to the path, then set \(D[i,\sigma]=f_\sigma^\gamma\).
In the other case, set \(D[i,\sigma] = \bot\).
If there is a valid signature \(\sigma\) for \(X_k\) such that \(D[k,\sigma]\neq \bot\), then the algorithm returns \emph{yes}, in the other case, it returns \emph{no}.

\begin{lemma*}\label{lem:pw3_midotherunique}
   Let \(i\geq 2\) and let \(\sigma\) be a valid signature for \(X_i\) that induces a non-trivial midpart and another non-trivial path with \(\ell(\sigma) < i\).
   Then there is at most one signature \(\gamma\) for \(X_{i-1}\) such that \(D[i-1,\gamma]\neq \bot\).
\end{lemma*}
\begin{proof}
Without loss of generality we may assume that the other path is a prefix.
If there was more than one compatible signature \(\gamma\) for \(\sigma\) with \(D[i-1,\gamma]\neq \bot\), then
there are at least two sequences \(\sigma =\sigma_i^1, \sigma_{i-1}^1,\dots\allowbreak, \sigma_{\ell(\sigma)+1}^1,  \sigma_{\ell(\sigma)}^1 = \tau(\sigma)\) and \(\sigma = \sigma_i^2, \sigma_{i-1}^2,\dots\allowbreak, \sigma_{\ell(\sigma)+1}^2,  \sigma_{\ell(\sigma)}^2 = \tau(\sigma)\), such that \(D[j,\sigma_j^1]\neq\bot\) and \(D[j,\sigma_j^2]\neq \bot\) for all \(\ell(\sigma)\leq j \leq i-1\).

Applying \cref{lem:pw3_onlyone} with \(i= \ell(\sigma)\) implies that \(\sigma_{\ell(\sigma)+1}^1 = \sigma_{\ell(\sigma)+1}^2\).
Applying this inductively yields \(\sigma_j^1 = \sigma_j^2\) for all \(j\) with \(\ell(\sigma)\leq j \leq i-1\) and the lemma follows.
 \end{proof}

\begin{theorem}\label{thm:pw3}
    POHPP in graphs of \param{pathwidth} at most \(3\) can be solved in \(\O(n^3)\) time.  
\end{theorem}

\begin{proof}
 As in each bag one vertex is forgotten, there are \(\O(n)\) possible bags.
 For each bag, there are \(\O(n)\) valid signatures as there are only $\O(1)$ valid choices of $\cP(\sigma)$, at most $\O(n)$ choices for $\ell(\sigma)$ and $\O(1)$ choices for $\tau(\sigma)$.

 There are two cases for these signatures of $X_i$. First, we consider those signatures $\sigma$ where $\cP(\sigma)$ contain a non-trivial midpart and another non-trivial path and $\ell(\sigma) < i$. There are $\O(n)$ of these signatures for a bag $X_i$. By \cref{lem:pw3_midotherunique}, there is only one signature $\gamma$ of \(X_{i-1}\) for which we have to check if \(f_\sigma^\gamma\) contradicts \(\pi\). 
 Finding  \(\gamma\) takes \(\O(n)\) time. Checking if it contradicts \(\pi\) takes an additional \(\O(n)\) time. We only have to find out where the predecessors and successors of the newly introduced vertex $w$ lie in the path mapping. Thus, we need a total of \(\O(n)\) time to compute \(D[i,\sigma]\). As there are $\O(n)$ many of these signatures, we can compute their $D$-values in $\O(n^2)$ total time.

 Now let us consider the other signatures. There are only $\O(1)$ of them for a bag $X_i$. There might be $\O(n)$ valid signatures of $X_{i-1}$ that are compatible for $\sigma$. With the same argument as above, we can check for any of those $\O(n)$ signatures whether the resulting path mapping contradicts $\pi$ in $\O(n)$ time. Thus, we need $\O(n^2)$ time in total to compute the $D$-values of one such signature $\sigma$ of $X_i$. As there are $\O(1)$ of them, we need $\O(n^2)$ time to compute the $D$-values of all these signatures of $X_i$.

 Thus, processing one bag costs $\O(n^2)$ time. As there are $\O(n)$ bags, our algorithm needs $\O(n^3)$ time in total.

For the correctness, if \(\cP(\sigma)\) contains only one non-trivial path then by \cref{obs:pw:unique_partition} there is only one possible path partition.
     If it contains a prefix and a suffix, then the vertices in \(V \setminus V_i\) are added between the end vertex of the prefix and the start vertex of the suffix.
    Let \(f_\sigma, f'_\sigma\) be a two path mappings for \(\sigma\) that both do not contradict \(\pi\).
    Then \(f_\sigma\) and \(f'_\sigma\) only differ in vertices that are incomparable to all vertices in \(V\setminus V_i\).
    Thus, if there is \(\pi\)-extending Hamiltonian path that induces the path partition \(f_\sigma\), then there is also one with path partition \(f'_\sigma\) and it suffices to store only one of them.
    Finally, assume that \(\sigma\) induces a non-trivial midpart and another non-trivial path.    
   By \cref{lem:pw3_midotherunique}, there is at most one compatible signature \(\gamma\) for \(\sigma\) for which \(D[i-1,\gamma]\neq\bot\) and, thus, only one candidate for a path mapping.
\end{proof}

\begin{corollary*}
    POHPP in grids of height at most \(3\) can be solved in \(O(n^3)\) time.
\end{corollary*}

\begin{theorem*}
    Given an $n$-vertex graph $G$ of pathwidth at most~$3$, we can solve  MinPOHPP in time $\O(n^3)$.
\end{theorem*}
\begin{proof}
    The algorithm described above can be modified to store the weight that stems from a path partition, together with the partition. 
    If there are multiple possible path partitions, the one that induces the smallest weight is stored.
\end{proof}

\subsubsection{Treewidth at most 2}\label{subsec:treewidth2}
In this section we consider graphs that have \param{treewidth} at most \(2\).
We reuse some ideas from \cref{subsec:pathwidth3} together with some additional observations.
Let \(\mathcal{T}=(\mathcal{X}, E)\) with \(\mathcal{X}=\{X_1,\dots, X_k\}\) be a tree decomposition of width \(2\) for \(G\). 
Without loss of generality, we assume that every bag contains exactly \(3\) vertices and that \(\mathcal{T}\) is a binary tree. 
The nodes with degree \(2\) in \(\mathcal{T}\) are called \emph{join nodes} and the other nodes \emph{exchange nodes}.
Let \(t\) be a join node with children \(t_1, t_2\). Then \(X_t = X_{t_2} = X_{t_2}\).
We also assume that no leaf of \(\mathcal{T}\) is a child of a join node.
For the exchange nodes, we make the same assumption as for the vertex sets in \cref{subsec:pathwidth3}, namely that they differ in exactly two vertices. Let \(t_1\) be an exchange node with child \(t_2\), then the unique vertex in \(X_{t_1}\setminus X_{t_2}\) is introduced in \(X_{t_1}\) and the vertex in \(X_{t_2}\setminus X_{t_1}\) is forgotten in \(X_{t_1}\).
We can compute a nice tree decomposition of width at most \(2\) in linear time \cite{bodlaender1996linear-time, furer2016faster}. Given a nice tree decomposition of width at most \(2\), a tree decomposition of the form above, can be found in in linear time.

For a node \(t\), let \(V_t\) be the union of \(X_{t'}\) for all \(t'\) in the subtree of \(\mathcal{T}\) rooted at \(t\).
Furthermore, let \(G_t\) be the induced subgraph on \(V_t\).
Let \(\cP\) be an ordered Hamiltonian path in \(G\) and let \(\cP_t= G_t \cap \cP\). We reuse the definitions of start vertex, end vertex, interior vertex, terminal vertex and trivial paths given in section before.

\begin{lemma*}\label{lem:tw2_form_non-join}
    Let \(\cP\) be a Hamiltonian path in \(G\) and \(t\) be a node whose parent is an exchange node. Then \(\cP_t\) has one of the following forms:
    \begin{enumerate}
        \item One non-trivial path together with at most two trivial paths
        \item A prefix and a suffix
    \end{enumerate}
\end{lemma*}
\begin{proof}
We can use \cref{lem:pw3_storinternal} since \(t\) behaves as a node in a path decomposition. 
Thus, by the same arguments as in the proof of \cref{lem:pw3_form}, \(\cP_t\) cannot contains more than two non-trivial paths.
    
   As there are only three vertices in each bag, we can additionally exclude the case of a prefix or suffix together with a midpart.
    Assume that \(\cP_t\) consists of a prefix and a midpart and let \(v\) be the vertex forgotten in the parent of \(t\).
    As all vertices in \(X_t\) are non-global terminal vertices of a path in \(\cP_t\), by \cref{lem:pw3_storinternal} there is no vertex that can be forgotten in the parent of \(t\).
\end{proof}

\begin{lemma*}\label{lem:tw2_form_join}
Let \(\cP\) be a Hamiltonian Path and let \(t\) be a node whose parent is a join node. Then \(\cP_t\) has one of the forms stated in \cref{lem:tw2_form_non-join} or it contains a midpart together with either a prefix or a suffix. 
\end{lemma*}
\begin{proof}
The statement follows directly from the fact that there are at most three terminal vertices in the bag.
\end{proof}

\begin{lemma*}\label{lem:tw2_mid+other}
    Let \(\cP\) be an ordered Hamiltonian path and let \(t\) be a join node, such that \(\cP_t\) contains a midpart and another path. 
    Let \(t'\) be the parent of \(t\) and \(t_1,t_2\) be the children of \(t\).
    Then \(\cP_{t'}\) has one of the forms stated in \cref{lem:tw2_form_non-join}.
    Furthermore, \(\cP_{t_1}\) and \(\cP_{t_2}\) contain exactly one non-trivial path each.
\end{lemma*}
\begin{proof}
    Without loss of generality, we assume that $\cP_t$ contains a prefix.
    By the form of the tree decomposition, \(t_1, t_2\) are not leaves of \(\mathcal{T}\) and thus there is at least one vertex in \(V_{t_i}\setminus X_{t_i}\). 
    This implies that \(\cP_{t_i}\) contains at least one non-trivial path for each \(i\in \{1,2\}\).

Now assume that for one of the children, say \(t_1\), \(\cP_{t_1}\) contains two non-trivial paths \(P_1\) and \(P_2\) and \(\cP_{t_2}\) contains a path \(P_3\). 
Only one of these \(P_i\) can be a prefix and none of these paths can be a suffix. So at least two of the three paths are midparts. 
As these midparts cannot have the same terminal vertices, in \(X_t\) one of the midparts joins the other two paths, a contradiction to the assumption of the form of \(\cP_t\).

    Now consider the parent \(t'\) of \(t\).
If $t'$ is an exchange node, then the statement follows from \cref{lem:tw2_form_non-join}. 
Thus, we may assume that \(t'\) is a join node. 
    Let $t''$ be the sibling of \(t\). $\cP_{t''}$ contains at least one non-trivial path by a similar argument as above. If one of these paths is a suffix, then, by \cref{lem:tw2_form_join}, \(\cP_{t'}\) contains a prefix and a suffix. Otherwise, the midpart of $P_{t''}$ joins the two paths of $\cP_{t}$ to one path in $\cP_{t'}$.
 \end{proof}

As in \cref{subsec:pathwidth3}, we define the signature \(\sigma\) of a bag \(X_t\) as a mapping of the vertices to the possible types of vertices and paths.
Additionally, if the signature induces midpart and another path, \(\sigma\) stores an additional bit \(d=\{L,R\}\).
This bit is used to encode if the other path came from the left or from the right subtree.
A signature is \emph{valid} if it has one of the forms stated by \cref{lem:tw2_form_non-join} and \cref{lem:tw2_form_join}.
Let \(\sigma\) be a valid signature for a node \(t\).
If \(t\) is an exchange node, we call a signature \(\gamma\) for its child \emph{compatible} if \(\gamma\) is valid and there is a way to extend the paths induced by \(\gamma\) with the vertex introduced in \(X_t\).
Similarly, if \(t\) is a join node, we call the signatures \(\gamma_1, \gamma_2\) for its left and right children respectively compatible to \(\sigma\) if joining the paths induced by \(\gamma_1, \gamma_2\) gives the paths induced by \(\sigma\).
In particular, if \(\sigma\) induces a midpart and, e.g., a prefix and \(d=L\), then \(\gamma_1\) has to induce a prefix. In the other case (\(d=R\)), \(\gamma_2\) has to induce a prefix.

\begin{lemma*}\label{lem:tw2_unique_compatible}
Let \(\sigma\) be a signature for a join node such that \(\sigma\) induces a midpart and another path. 
Then there is only one pair \(\gamma_1, \gamma_2\) of compatible signatures for the children \(t_1,t_2\) of \(t\).
\end{lemma*}
\begin{proof}
Without loss of generality assume that the other path is a prefix.
In \(\sigma\) there are no interior vertices. Let \(v\in X_t\) be the vertex assigned to the prefix and \(t_p\) be the child indicated by \(d(\sigma)\). 
Then the compatible signature for \(t_p\) assigns \(v\) to be the end vertex of a prefix and the other vertices to trivial paths. 
Consequently, the vertices in \(X_t\setminus \{v\}\) are the start and end vertices of the midpart, as assigned by \(\sigma\).
\end{proof}

Given a signature \(\sigma\) for \(X_t\), a \emph{path mapping} \(f_\sigma\) assigns each vertex \(v\in V_t\) to one of the paths induced by \(\sigma\).
Again, a path mapping \emph{contradicts} \(\pi\) if there is no way to append or prepend the vertices in \(V\setminus V_t\) to the paths without violating the constraints given by \(\pi\).

We now have all the definitions to define the dynamic program for POHPP. 
For each bag \(X_t\) and each valid signature \(\sigma\) for \(X_t\), define the subproblem \(D[t,\sigma]\).
If there is a path mapping \(f_\sigma\) for a signature \(\sigma\) of \(X_t\), then \(D[t,\sigma]=f_\sigma\), otherwise \(D[t,\sigma] = \bot\).
The table is filled bottom up along the tree. 
For a leaf there is only one possible path partition where it has to be checked if it contradicts \(\pi\).

For an exchange node, the same algorithm as in \cref{subsec:pathwidth3} can be used. 
Now consider a join node \(t\) with signature \(\sigma\).
Iterate over all compatible signatures \(\gamma_1,\gamma_2\) for its children and consider the uniquely defined path partition induced by this triple. 
If there is one pair of compatible signatures that induce a path partition \(f_\sigma\) that does not contradict \(\pi\), set \(D[t,\sigma]\) to this value, otherwise set \(D[t,\sigma]= \bot\).

\begin{theorem}\label{thm:tw2}
    (Min)POHPP in graphs of \param{treewidth} at most \(2\) can be solved in \(\O(n^2)\) time.
\end{theorem}
\begin{proof}
There are \(\O(n)\) bags as every vertex is forgotten at most once. 
Furthermore, for each bag, there are \(\O(1)\) possible valid signatures. 
Consequentially, there can be only \(\O(1)\) compatible signatures to check for each node in linear total time for each entry \(D[i,\sigma]\).
As in the proof of \cref{thm:pw3}, it is enough to store one path partition if the signature induces only one (non-trivial) path or a prefix and a suffix.

Furthermore, by \cref{lem:tw2_mid+other}, if there is a signature \(\sigma\) for a node \(t\) that induces a non-trivial midpart and a non-trivial prefix, 
\cref{lem:tw2_form_non-join} and \cref{lem:tw2_form_join} imply that this node is a join node and by  \cref{lem:tw2_unique_compatible} there is only one pair \(\gamma_1,\gamma_2\) of compatible signatures for \(\sigma\).
Furthermore, by \cref{lem:tw2_mid+other}, the parent of \(t\) induces either one path or a prefix and a suffix if \(\sigma\) is compatible with a signature.
Thus, similar as stated in the proof of \cref{thm:pw3}, the information about which side contained the start vertex is not needed further above in the tree and can safely be ignored.

By storing the path mapping with minimum weight if there is more than one candidate, the algorithm above extends to  MinPOHPP.
\end{proof}

\subparagraph*{A note on treewidth 3} 
When seeing \cref{thm:pw3} and \cref{thm:tw2}, one might hope to combine and extend the algorithms to give an algorithm for \param{treewidth} at most \(3\). 
This however seems to pose a bigger challenge as \cref{lem:tw2_unique_compatible} is not true anymore. 
Thus, the signature for nodes with a midpart and another path have to maintain more than local information about the origin of each path.
In \cref{thm:pw3}, this global information was efficiently encoded by relying on the linear structure of the path decomposition.
This seems not to be that easy for graphs of \param{treewidth} \(3\), leaving an interesting avenue for further research.

\subsubsection{Treedepth}

\begin{observation*}\label{obs:treedepth}
    If a traceable graph has \param{treedepth} $k$, then $G$ has less than $2^k$ vertices.
\end{observation*}

\begin{proof}
   \param{Treedepth} is monotone {\cite[Lemma~6.2]{nesetril2012sparsity}}, i.e, considering subgraphs does not increase the \param{treedepth}. Furthermore, the \param{treedepth} of a path with $n$ vertices is $\lceil \log_2(n+1) \rceil$  {\cite[Equation~6.2]{nesetril2012sparsity}}. Hence, if a graph has \param{treedepth} $k$, then it does not contain a path with $2^k$ vertices. 
\end{proof}

This observation implies that for traceable graphs with bounded \param{treedepth} all other graph width parameters are bounded (see \cref{fig:parameters}). Since we can solve  MinPOHPP in $\O(2^n)$ time on graphs with $n$ vertices~\cite{beisegel2024computing}, the following running time bound holds.

\begin{theorem}\label{corol:treedepth-fpt}
     MinPOHPP can be solved in $\O(2^{2^k})$ time on graphs of \param{treedepth}~$k$. 
\end{theorem}

\section{Distance Parameters to Sparse Graph Classes}\label{sec:d2p}
As POHPP is para-\NP-complete for \param{treewidth}, we next consider parameters that form upper bounds of \param{treewidth}.

\subsection{Hardness}

We start by showing that even for very restricted \param{distance to \G} parameters, there is no hope for \FPT{} algorithms.

\begin{theorem}\label{thm:w1-d2p}
     POHPP is \W-hard when parameterized by \param{distance to path}.
\end{theorem}

\begin{proof}

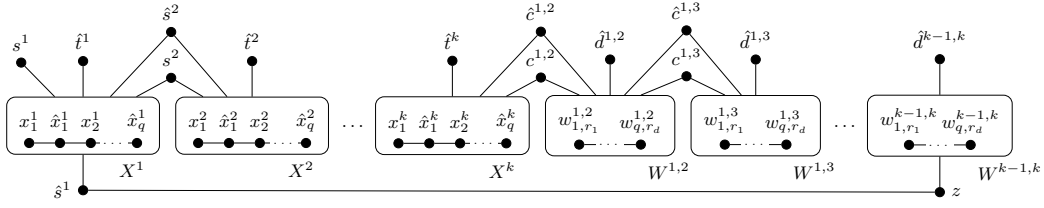
\begin{figure}
    \centering
    \resizebox{\textwidth}{!}{
    \begin{tikzpicture}[xscale=1, clique/.style={draw, rounded corners}]
\footnotesize

        \node[clique, label=-45:$X^1$] (X1) at (0,0) {
        \begin{tikzpicture}
        \node[vertex, label=90:$x^1_1$] (x11) at (0,0) {};
        \node[vertex, label=90:$\hat{x}^1_1$] (hx11) at (0.5,0) {};
        \node[vertex, label=90:$x^1_2$] (x12) at (1,0) {};
        \node[vertex, label=90:$\hat{x}^1_q$] (x1q) at (1.75,0) {};
        \draw (x11) -- (hx11) -- (x12) -- (x1q);
        \node[fill=white, rectangle, rounded corners=false, inner sep=1pt] at (1.375,0) {\tiny $\dots$};
        \end{tikzpicture}
        };

        \node[clique, label=-45:$X^2$] (X2) [right=0.25cm of X1] {
        \begin{tikzpicture}
        \node[vertex, label=90:$x^2_1$] (x11) at (0,0) {};
        \node[vertex, label=90:$\hat{x}^2_1$] (hx11) at (0.5,0) {};
        \node[vertex, label=90:$x^2_2$] (x12) at (1,0) {};
        \node[vertex, label=90:$\hat{x}^2_q$] (x1q) at (1.75,0) {};
        \draw (x11) -- (hx11) -- (x12) -- (x1q);
        \node[fill=white, rectangle, rounded corners=false, inner sep=1pt] at (1.25,0) {\tiny $\dots$};
        \end{tikzpicture}
        };

        \node [right=0.1cm of X2] {$\dots$};

        \node[clique, label=-45:$X^k$] (Xk) [right=0.75cm of X2] {
        \begin{tikzpicture}
        \node[vertex, label=90:$x^k_1$] (x11) at (0,0) {};
        \node[vertex, label=90:$\hat{x}^k_1$] (hx11) at (0.5,0) {};
        \node[vertex, label=90:$x^k_2$] (x12) at (1,0) {};
        \node[vertex, label=90:$\hat{x}^k_q$] (x1q) at (1.75,0) {};
        \draw (x11) -- (hx11) -- (x12) -- (x1q);
        \node[fill=white, rectangle, rounded corners=false, inner sep=1pt] at (1.25,0) {\tiny $\dots$};
        \end{tikzpicture}
        };

        \node[vertex, label=90:$s^1$] (s1) [above left=0.5cm and -0.3cm of X1] {};
        \node[vertex, label=90:$\hat{t}^1$] (t1) [above=0.5cm of X1] {};
        \node[vertex, label=180:$\hat{s}^1$] (hs1) [below=0.5cm of X1] {};
        \node[vertex, label=90:$s^2$] (s2) [above right=0.25cm and 0.125cm of X1] {};
        \node[vertex, label=90:$\hat{s}^2$] (hs2) [above right=1cm and 0.125cm of X1] {};
        \node[vertex, label=90:$\hat{t}^2$] (t2) [above=0.5cm of X2] {};
        \node[vertex, label=90:$\hat{t}^k$] (tk) [above=0.5cm of Xk] {};
    
        \draw (t1) -- (X1);
        \draw (t2) -- (X2);
        \draw (tk) -- (Xk);
        \draw (hs1) -- (X1);
        \draw (s1) -- (X1);
        \draw (s2) -- (X1);
        \draw (hs2) -- (X1);
        \draw (s2) -- (X2);
        \draw (hs2) -- (X2);

        \node[clique, label=-45:$W^{1,2}$] (W12) [right=0.25cm of Xk] {
        \begin{tikzpicture}
        \node[vertex, label=90:$w^{1,2}_{1,r_1}$] (w11) at (0,0) {};
        \node[vertex, label=90:$w^{1,2}_{q,r_d}$] (w1q) at (1,0) {};
        \draw (w11) -- (w1q);
        \node[fill=white, rectangle, rounded corners=false, inner sep=1pt] at (0.35,0) {\tiny $\dots$};
        \end{tikzpicture}
        };

        \node[clique, label=-45:$W^{1,3}$] (W13) [right=0.25cm of W12] {
        \begin{tikzpicture}
        \node[vertex, label=90:$w^{1,3}_{1,r_1}$] (w11) at (0,0) {};
        \node[vertex, label=90:$w^{1,3}_{q,r_d}$] (w1q) at (1,0) {};
        \draw (w11) -- (w1q);
        \node[fill=white, rectangle, rounded corners=false, inner sep=1pt] at (0.35,0) {\tiny $\dots$};
        \end{tikzpicture}
        };

        \node [right=0.1cm of W13] {$\dots$};

        \node[clique, label=-45:$W^{k-1,k}$] (Wk-1k) [right=0.75cm of W13] {
        \begin{tikzpicture}
        \node[vertex, label=90:$w^{k-1,k}_{1,r_1}$] (w11) at (0,0) {};
        \node[vertex, label=90:$w^{k-1,k}_{q,r_d}$] (w1q) at (1,0) {};
        \draw (w11) -- (w1q);
        \node[fill=white, rectangle, rounded corners=false, inner sep=1pt] at (0.35,0) {\tiny $\dots$};
        \end{tikzpicture}
        };

        \node[vertex, label=90:$\hat{d}^{1,2}$] (d12) [above=0.5cm of W12] {};
        \node[vertex, label=90:$\hat{d}^{1,3}$] (d13) [above=0.5cm of W13] {};
        \node[vertex, label=90:$\hat{d}^{k-1,k}$] (dk-1k) [above=0.5cm of Wk-1k] {};
        \node[vertex, label=90:$c^{1,2}$] (c12) [above right=0.25cm and 0.125cm of Xk] {};
        \node[vertex, label=90:$\hat{c}^{1,2}$] (hc12) [above right=1cm and 0.125cm of Xk] {};
        \node[vertex, label=90:$c^{1,3}$] (c13) [above right=0.25cm and 0.125cm of W12] {};
        \node[vertex, label=90:$\hat{c}^{1,3}$] (hc13) [above right=1cm and 0.125cm of W12] {};
        \node[vertex, label=0:$z$] (z) [below=0.5cm of Wk-1k] {};
        
        \draw (d12) -- (W12);
        \draw (d13) -- (W13);
        \draw (dk-1k) -- (Wk-1k);
        \draw (c12) -- (Xk);
        \draw (hc12) -- (Xk);
        \draw (c12) -- (W12.135);
        \draw (hc12) -- (W12.115);
        \draw (c13) -- (W12.45);
        \draw (hc13) -- (W12.65);
        \draw (c13) -- (W13.135);
        \draw (hc13) -- (W13.115);
        \draw (z) -- (Wk-1k);
        \draw (z) -- (hs1);
    \end{tikzpicture}
    }
    \caption{Construction for \cref{thm:w1-d2p,thm:w1-lfm}. If a vertex is adjacent to a box, then the vertex is adjacent to all vertices in that box. In the proof of \cref{thm:w1-d2p}, the ends of consecutive subpaths of the $X^i$ and $W^{i,j}$ are adjacent, in the proof of \cref{thm:w1-lfm} they are not adjacent.}
    \label{fig:w1-d2p}
\end{figure}

    We reduce the \textsc{Multicolored Clique Problem} to  POHPP parameterized by \param{distance to path}. Let $G$ be an instance of the MCP. The vertex set of $G$ consists of $k$ color classes $V_1 \cup \dots \cup V_k$ and every color class consists of vertices $v^i_1, \dots, v^i_q$.

    We construct a graph $G'$ as follows (see \cref{fig:w1-d2p}). The vertex set is partitioned in the following gadgets:
    \begin{description}
    \item[Selection Gadget] For every color class $i \in [k]$, we have a selection gadget consisting of a vertex $\hat{t}^i$ and a path $X_i$ that contains for every $v^i_p \in V_i$ a vertex $x^i_p$. These vertices are ordered by their index $p$ and after every vertex $x^i_p$ (and, thus, before $x^i_{p+1}$) there is a vertex $\hat{x}^i_p$ in that path. The vertex $\hat{t}^i$ is adjacent to all vertices in the path $X^i$.
    \item[Verification Gadget] For each pair $i,j \in [k]$ with $i < j$, we have a verification gadget consisting of a vertex $\hat{d}^{i,j}$ and a path $W^{i,j}$ containing for every edge $v^i_pv^j_r \in E(G)$ a vertex $w^{i,j}_{p,r}$. We do not fix an ordering of these vertices in that path. The vertex $\hat{d}^{i,j}$ is adjacent to all vertices in the path $W^{i,j}$.
    \end{description}

    Next we describe how these gadgets are connected to each other. We order the subpaths described in the selection and verification gadgets as follows: \[X^1, X^2, \dots, X^k, W^{1,2}, W^{1,3}, \dots, W^{1,k}, W^{2,3}, W^{2,4}, \dots, W^{k-1,k}.\]
    The last vertex of one of the paths is adjacent to the first vertex of the succeeding path.
    Therefore, these paths form one large path $\Psi$ in $G'$.
    Additionally, we have the following vertices in $G'$.
    First, we have $s^1$ and $\hat{s}^1$ that are adjacent to all vertices in $X^1$. For every $i \in [k]$ with $i > 1$, we have a vertex $s^i$ and a vertex $\hat{s}^i$ that both are adjacent to all vertices in $X^i$ and $X^{i-1}$. 
    For all $i,j \in [k]$ with $i < j$, there are vertices $c^{i,j}$ and $\hat{c}^{i,j}$ that are adjacent to all vertices in $W^{i,j}$ and to all vertices in the predecessor path of $W^{i,j}$ in the ordering described above.
    Finally, we have a vertex $z$ that is adjacent to all vertices in $W^{k-1,k}$ and to $\hat{s}^1$. Observe that the graph $G' - \Psi$ contains $3k + 3\binom{k}{2} + 1$ vertices. Therefore, the \param{distance to path} of $G'$ is $\O(k^2)$.

    We define $Y$ as the set containing all the vertices defined above that have a hat in their name. The partial order $\pi$ is the reflexive and transitive closure of the following constraints:
    \begin{enumerate}[(P1)]
        \item $s^1 \prec v$ for all $v \in V(G') \setminus \{s^1\}$,\label{p0}
        \smallskip
        \item $z \prec y$ for all $y \in Y$,\label{p6}
        \smallskip
        \item $x^i_p \prec w^{i,j}_{p,r}$ for all $i,j \in [k]$ with $i < j$ and all $p,r \in [q]$ with $v^i_pv^j_r \in E(G)$\label{p4}
        \smallskip
        \item $x^j_r \prec w^{i,j}_{p,r}$ for all $i,j \in [k]$ with $i < j$ and all $p,r \in [q]$ with $v^i_pv^j_r \in E(G)$\label{p5}
    \end{enumerate}

    \begin{claim*}\label{claim:w1-d2p1}
        If $G$ has a multicolored clique $\{v^1_{p_1}, \dots, v^k_{p_k}\}$, then $G'$ has a $\pi$-extending Hamiltonian path $\cP$.
    \end{claim*}

    \begin{claimproof}
        We start in $s^1$. Now we visit $x^1_{p_1}$ and go to $s^2$. Then we go to $x^2_{p_2}$. We repeat this process until we reach $x^k_{p_k}$. Next we visit $c^{1,2}$ and then $w^{1,2}_{p_1,p_2}$. Note that this is possible since both $x^1_{p_1}$ and $x^2_{p_2}$ are already visited and, thus, the constraints given in \pef{p4} and \pef{p5} are fulfilled. We then go to $c^{1,3}$ and visit $w^{1,3}_{p_1, p_3}$. We repeat this procedure until we reach $w^{k-1,k}_{p_{k-1}, p_k}$. Then we go to $z$ and then to $\hat{s}_1$. Now we traverse the path $X^1$. When we reach $\hat{x}^1_{p_1 - 1}$, we cannot visit the next vertex on the path as this vertex has already been visited. Therefore, we use $\hat{t}^1$ to jump over that vertex in $X^1$. When we have traversed $X^1$ completely, we visit $\hat{s}^2$ and then traverse $X^2$ in the same way as $X^1$. We repeat this procedure for all $X^i$ and also for all $W^{i,j}$, where we use $\hat{d}_{i,j}$ to jump over visited vertices.
    \end{claimproof}

    \begin{claim*}\label{claim:w1-d2p2}
        If there is a $\pi$-extending Hamiltonian path $\cP$ in $G'$, then there is a multicolored clique $\{v^1_{p_1}, \dots v^k_{p_k}\}$ in $G$.
    \end{claim*}

    \begin{claimproof}
        The path $\cP=(a_1,a_2,\dots a_n)$ has to start in $a_1=s^1$, due to \pef{p0}. Then \(a_2\) has to be some vertex $x^1_{p_1}$ for some $p_1 \in [q]$ since $s^1$ has no other neighbors. The next vertex (\(a_3\)) cannot be a neighbor of $x^1_{p_1}$ on the path $\Psi$, due to \pef{p6}. Hence, \(a_3\) has to be $s^2$. Its successor \(a_4\) either could be a vertex $x^2_{p_2}$ or a vertex $x^1_{p'}$. In the latter case, all the unvisited neighbors of $x^1_{p'}$ are not allowed to be taken next as they are forced to be to the right of $z$ by \pef{p6}. Therefore, $a_4 = x^2_{p_2}$ for some $p_2 \in [q]$. This argument can be repeated for every $i \in [k]$. Thus, the subpath of $\cP$ between $s^1$ and $c^{1,2}$ contains for any $i \in [k]$ the vertex $s^i$ and exactly one vertex $x^i_{p_i}$ for some $p_i \in [q]$. We claim that the vertices $C = \{v^1_{p_1}, \dots, v^k_{p_k}\}$ form a multicolored clique in $G$. 

        Using the same argument as above, the path $\cP$ has to go from $c^{1,2}$ to some vertex $w^{1,2}_{p,r}$. Due to \pef{p5}, $p$ must be equal to $p_1$. By \pef{p4}, $r$ has to be $p_2$. Therefore, $v^1_{p_1}$ and $v^2_{p_2}$ are adjacent. This observation also implies that the next vertex cannot be on $W^{1,2}$ but has to be $c^{1,3}$. Repeating this argumentation, we can show that $C$ in fact forms a multicolored clique in $G$.
    \end{claimproof}

    Obviously, the graph $G'$ can be constructed in $\O((kq)^2)$ time. Combining this with \cref{claim:w1-d2p1} and \cref{claim:w1-d2p2}, shows that $(G', \pi)$ is a valid \FPT{} reduction. This finalizes the proof of the theorem.
\end{proof}

If we remove the edges between succeeding subpaths of $\Psi$ in $G'$, the graph induced by the vertices of $\Psi$ is a linear forest whose components are modules in $G'$. We can observe that these edges are not used in a Hamiltonian path in the proof of \cref{thm:w1-d2p}.  Therefore, the construction of \cref{thm:w1-d2p} can also be used to show the following.

\begin{theorem*}\label{thm:w1-lfm}
     POHPP is \W-hard when parameterized by \param{distance to linear forest modules}.
\end{theorem*}

\subsection{Algorithms}

We start this section with an observation that follows from the fact that deleting a set of $k$~vertices from a traceable graphs produces at most $k+1$ components.

\begin{observation*}
    If a traceable graph $G$ has \param{vertex cover number}~$k$, then $G$ has at most $2k+1$ vertices.
\end{observation*}

Similar as \cref{obs:treedepth}, this fact implies that traceable graphs with bounded \param{vertex cover number} have bounded values for all other graph width parameters (see \cref{fig:parameters}). Furthermore, the observation directly gives a linear kernel and an \FPT{} algorithm.

\begin{theorem}\label{cor:fpt-vc}
    MinPOHPP parameterized by the \param{vertex cover number}~$k$ has a kernel with at most $2k+1$ vertices that can be computed in $k^{\O(1)}$ time. Furthermore,  MinPOHPP can be solved in $\O(4^{k})$ time on graphs of \param{vertex cover number}~$k$.
\end{theorem}

Next, we extend the \FPT{} result to the \param{feedback edge set number}, also called \param{cycle rank}. We use the following result.

\begin{lemma*}[Demaine et al.~\cite{demaine2019reconfiguring}]\label{lemma:fes-paths}
    An $n$-vertex graph with \param{feedback edge set number}~k has at most $2^k \cdot \binom{n}{2}$ different paths.
\end{lemma*}

Enumerating all possible paths leads to the following running time for  MinPOHPP parameterized by \param{feedback edge set number}.

\begin{theorem}\label{thm:fes}
     MinPOHPP can be solved in $\O(2^k \cdot n^{\O(1)})$ time on an $n$-vertex graph with \param{feedback edge set number}~$k$.
\end{theorem}

We can adapt this algorithm to present an \XP{} algorithm for \param{feedback vertex set number}.

\begin{theorem}
     MinPOHPP can be solved in $n^{\O(k)}$ time on an $n$-vertex graph with \param{feedback vertex set number}~$k$.
\end{theorem}

\begin{proof}
    To find a feedback vertex set $W$ of size $k$ you can use your favorite \FPT{} algorithm (see, e.g., \cite{cao2015feedback,kociumaka2014faster}) or you find it in $n^{\O(k)}$ time by enumerating all vertex subsets of size $k$. Fix one of these sets $W$. For every vertex in $W$, we choose one predecessor and one successor (or we decide that the vertex is the first or the last vertex of our Hamiltonian path). We call these decisions a \emph{vertex choice}. There are $n^{\O(k)}$ many vertex choices. 
    
    For every of those vertex choices, we construct the graph $G'$ as follows. We delete all vertices of $W$ from $G$. Consider the pairs of predecessors and successors of the vertices in $W$. If they form a path starting and ending in vertices of $G - W$, we add an edge to $G'$ connecting the first vertex $s$ of this path with the last vertex $t$ of the path.
    As for every of the \(k\) vertices in \(W\) there is at most one edge in \(G'\), the resulting graph has \param{feedback edge number} \(k\).
    
    We update the partial order $\pi$ accordingly, i.e., for any inner vertex on this path, we make its predecessors in $\pi$ to predecessors of $s$ and its successors in $\pi$ to successors of $t$. If we have chosen some vertex of $W$ as start vertex of the Hamiltonian path, then we add constraints to the partial order that makes $t$ (the first successor of the start vertex in $G-W$) the start vertex. Equivalently, if a vertex of $W$ is chosen to be the end vertex of the Hamiltonian path, then we make the vertex $s$ (the last predecessor of the end vertex in $G-W$) the end vertex of the partial order. We call the resulting updated partial order $\pi'$.
    
   It is easy to see that a $\pi$-extending Hamiltonian path of $G$ that follows our vertex choice directly maps to a $\pi'$-extending Hamiltonian path of $G'$ that traverses all the added edges in the order that is implied by chosen predecessors and successors. We now apply the enumerating algorithm given in \cref{thm:fes} that takes $\O(2^k \cdot n^{\O(1)})$ time. For every of the enumerated paths, we check whether it uses the added edges in the way it is implied by our vertex choice. If this is the note case, we ignore this path. The total running time of this algorithm is $n^{\O(k)}$.
\end{proof}

It has been shown in \cite{beisegel2024computing} that  MinPOHPP can be solved in $\O(n^2)$ time on outerplanar graphs. We have to leave it open whether we can extend this result to an \XP{} algorithm for \param{distance to outerplanar}. Nevertheless, we are able to extend this result to graphs with an embedding with at most $k$ inner vertices. 
Note that a graph might have \param{distance to outerplanar} \(1\) but a large number of inner vertices.
Using ideas similar to those of \cite{deineko2006traveling}, we can give an \FPT{} algorithm for MinPOHPP parameterized by $k$. The dynamic programming approach is an adaptation of the algorithm for outerplanar graphs given in~\cite{beisegel2024computing}. The key ingredient of this algorithm is the following lemma, which shows that any prefix of a Hamiltonian path will always be made up of some interval on the outer face, as well as some of the extra vertices.

\begin{lemma*}[Beisegel et al.~\cite{beisegel2024computing}]\label{lemma:op-interval}
    Let $G$ be a planar graph and let $C$ be a face of a plane embedding of $G$. Furthermore, let $(v_0, \dots, v_k)$ be the cyclic ordering of the vertices on $C$ and let $P$ be a prefix of a Hamiltonian path of $G$. Then $V(C) \cap V(P) = \emptyset$ or there exist $q,r \in \{0,\dots,k\}$ with $q \leq r$ such that either $V(C) \cap V(P) = \{v_q,\dots,v_r\}$ or $V(C) \cap V(P) = \{v_r,\dots,v_k,v_0,\dots,v_q\}$.
\end{lemma*}

\begin{theorem}
    Given a planar graph $G$ with a planar embedding $P$ with $k$ vertices that are not on the outer face, we can solve  MinPOHPP on $G$ in $\O(2^k k n^3)$ time.
\end{theorem}
\begin{proof}
    We only sketch the idea of the algorithm as it is a quite straightforward adaption of Algorithm~2 in~\cite{beisegel2024computing}. For details, we refer to this publication. First note that considering induced subgraphs does not increase the number of inner vertices of an embedding. Hence, by \cite{beisegel2024computing}, it suffices to deal with 2-connected graphs.
    Let $C = (v_0, \dots, v_\ell)$ be the outer face of the planar embedding $P$ and $W$ be the set of vertices that are not on the outer face. Since $G$ is 2-connected, the face $C$ forms a cycle.

    The idea of Algorithm~2 in~\cite{beisegel2024computing} is to use a dynamic programming approach. To this end, we consider tuples $(a,b,\omega)$, which represent prefixes of Hamiltonian paths. Here, $a$ and $b$ are the indices of the endpoints of the interval on $C$ associated with that prefix (see~\cref{lemma:op-interval}) and $\omega$ is either 1 or 2 depending on whether the prefix ends in $v_a$ or $v_b$. For the problem considered here, we simply need to adjust these to tuples of the form $(a,b,X,t)$, where $a$ and $b$ again represent the indices of the endpoints of the interval on $C$ associated with the prefix, $X \subseteq W$ is the set of inner vertices in the prefix and $t$ forms the endpoint of the prefix, respectively. There are at most $n^2$ intervals and at most $2^k$ subsets of $W$. The endpoint $t$ either has to be a vertex of $X$ or one of the vertices $v_a$ and $v_b$. Therefore, there are at most $(k+2)$ endpoints and the total number of tuples is bounded by $2^k (k+2) n^2$. Similar as in Algorithm~2 of \cite{beisegel2024computing}, we compute for every tuple $(a,b,X,t)$ the minimal cost $M(a,b,X,t)$ of an ordered path $\cP$ of $G$ fulfilling the following properties (or $\infty$ if no such path exists):
\begin{enumerate}[(i)]
  \item $\cP$ consists of the vertices in the interval between $v_a$ and $v_b$ of $C$ and in the set $X$,\label{cond:outer1}
  \item $\cP$ is a prefix of a linear extension of $\pi$,\label{cond:outer2}
  \item $\cP$ ends in $t$.\label{cond:outer3}
\end{enumerate} 
This is done inductively by the size of $\cP$. Let $a'$, $b'$, and $X'$ be the updated values if $t$ is removed from the potential prefix. We first have to check whether $t$ is minimal in $\pi$ if all the vertices of $[a',b']$ and $X'$ have been visited. This costs $\O(n)$ time. If this is not the case, we set the $M$-value to $\infty$. Otherwise, we check the $M$-values for all possible vertices $t'$ that are before $t$ in the prefix. Note that $t'$ can only be an element of $X'$ or one of the two vertices $v_{a'}$ and $v_{b'}$, due to \cref{lemma:op-interval}. For each choice of $t'$ we compute the value $M(a',b',X',t') + c(tt')$ and set the $M$-value of $(a,b,X,t)$ to the minimum of these values. Note that this can be done in $\O(k) \subseteq \O(n)$ time using an adjacency matrix.

Summing up, for every of the $\O(2^kkn^2)$ tuples we need $\O(n)$ time which leads to the overall running time of $\O(2^kkn^3)$. The proof of the correctness of the algorithm follows along the lines of the proof of Theorem~5.3 in \cite{beisegel2024computing}.
\end{proof}

\section{Non-Sparse Width Parameters}\label{sec:d2c}

All parameters considered so far are sparse, i.e., the parameter is unbounded for cliques. As POHPP is trivial on cliques, it makes sense to consider also non-sparse parameters. Note that the MinPOHPP is \NP-hard on cliques as it forms a generalization of TSP. Therefore, we will only consider POHPP in this section.

\subsection{Hardness}

As mentioned in the introduction, \textsc{Hamiltonian Path} can be solved in polynomial time when the \param{independence number}, that is the size of the largest independent set, is fixed~\cite{fomin2024hamiltonicity,jedlickova2024hamiltonian}. Since this parameter is a lower bound on the \param{clique cover number}, the same holds for this graph parameter. Here, we show that this result cannot be extended to  POHPP unless $\P = \NP$.

\begin{theorem}\label{thm:cc}
     POHPP is \NP-complete on graphs of clique cover number~2.
\end{theorem}

\begin{proof}
    We reduce 3-SAT to  POHPP on graphs of \param{clique cover number}~2. Let $\Phi$ be a formula with variables $x_1, \dots, x_n$ and clauses $c_1, \dots, c_m$. Let $\ell_i^1$, $\ell_i^2$, and $\ell_i^3$ be the (possibly negated) literals contained in $c_i$. The graph $G$ contains the following gadgets (see \cref{fig:cc} for an illustration):
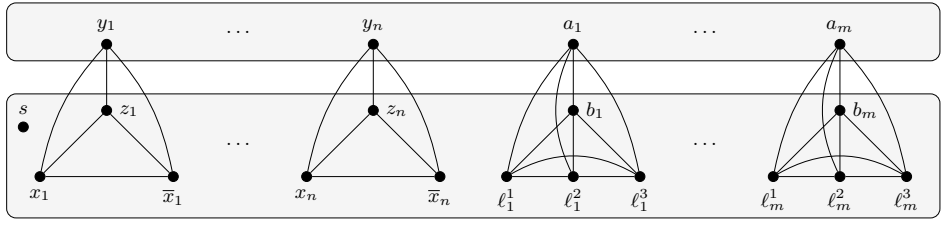
\begin{figure}
    \centering
    \resizebox{0.9\textwidth}{!}{
    \begin{tikzpicture}[yscale=1]
    \footnotesize
  
    \draw[rounded corners, fill=lightgray!15!white] (-0.5,-0.625) rectangle (13.5,1.25);
    \draw[rounded corners, fill=lightgray!15!white] (-0.5,1.75) rectangle (13.5,2.625);

    \node[vertex, label=90:$s$] at (-0.25,0.75) {};
    
    \node at (3,0.5) {$\dots$};
    \node at (3,2.1875) {$\dots$};
    \node at (10,0.5) {$\dots$};
    \node at (10,2.1875) {$\dots$};
    
    \node[vertex, label=-90:$x_1$] (xi) at (0,0) {};
    \node[vertex, label=-90:$\overline{x}_1$] (nxi) at (2,0) {};
    \node[vertex, label=0:$z_1$] (zi) at (1,1) {};
    \node[vertex, label=90:$y_1$] (yi) at (1,2) {};

    \draw (yi) -- (zi) -- (xi) -- (nxi) -- (zi);
    \draw (xi) to[bend angle=15, bend left] (yi);
    \draw (yi) to[bend angle=15, bend left] (nxi);

    \node[vertex, label=-90:$x_n$] (xi) at (4,0) {};
    \node[vertex, label=-90:$\overline{x}_n$] (nxi) at (6,0) {};
    \node[vertex, label=0:$z_n$] (zi) at (5,1) {};
    \node[vertex, label=90:$y_n$] (yi) at (5,2) {};

    \draw (yi) -- (zi) -- (xi) -- (nxi) -- (zi);
    \draw (xi) to[bend angle=15, bend left] (yi);
    \draw (yi) to[bend angle=15, bend left] (nxi);

    \node[vertex, label=-90:$\ell_1^1$] (l1) at (7,0) {};
    \node[vertex, label=-90:$\ell_1^2$] (l2) at (8,0) {};
    \node[vertex, label=-90:$\ell_1^3$] (l3) at (9,0) {};
    \node[vertex, label=0:$b_1$] (bj) at (8,1) {};
    \node[vertex, label=90:$a_1$] (aj) at (8,2) {};

    \draw (l1) -- (l2) -- (l3) -- (bj) -- (l2);
    \draw (l1) -- (bj) -- (aj);
    \draw[bend left] (l1) to (l3);
    \draw[bend angle=25, bend right] (aj) to (l2);
    \draw (aj) to[bend angle=15, bend right] (l1);
    \draw (aj) to[bend angle=15, bend left] (l3);

    \node[vertex, label=-90:$\ell_m^1$] (l1) at (11,0) {};
    \node[vertex, label=-90:$\ell_m^2$] (l2) at (12,0) {};
    \node[vertex, label=-90:$\ell_m^3$] (l3) at (13,0) {};
    \node[vertex, label=0:$b_m$] (bj) at (12,1) {};
    \node[vertex, label=90:$a_m$] (aj) at (12,2) {};

    \draw (l1) -- (l2) -- (l3) -- (bj) -- (l2);
    \draw (l1) -- (bj) -- (aj);
    \draw[bend left] (l1) to (l3);
    \draw[bend angle=25, bend right] (aj) to (l2);
    \draw (aj) to[bend angle=15, bend right] (l1);
    \draw (aj) to[bend angle=15, bend left] (l3);


\end{tikzpicture}
    }
    \caption{Construction for \cref{thm:cc}. The gray boxes form cliques.}
    \label{fig:cc}
\end{figure}    
\begin{description}
        \item[Variable Gadget] For every variable $x_i$, we have a clique consisting of the vertices $x_i$, $\overline{x}_i$, $y_i$, and $z_i$.
        \item[Clause Gadget] For every clause $c_i$, we have a clique consisting of the vertices $\ell_i^1$, $\ell_i^2$, $\ell_i^3$, $a_i$, and $b_i$.
    \end{description}
    Besides these vertices, we have a vertex $s$. Furthermore, we add edges such that the two sets 
    \[\{x_i, \overline{x}_i, z_i \mid i \in [n]\} \cup \{b_i, \ell_i^1, \ell_i^2, \ell_i^3 \mid i \in [m]\} \cup \{s\} \text{ and } \{y_i \mid i \in [n]\} \cup \{a_i \mid i \in [m]\}\] form cliques and, hence, $G$ has \param{clique cover number}~2. The partial order $\pi$ is the reflexive and transitive closure of the the following constraints:
    
    \begin{enumerate}[(P1)]
        \item $s \prec v$ for all $v \in V(G) \setminus \{s\}$,\label{cc-p1}
        \item $a_1 \prec b_1 \prec a_2 \prec b_2 \dots \prec a_m \prec b_m \prec y_1 \prec z_1 \prec y_2 \prec z_2 \prec \dots \prec y_n \prec z_n$,\label{cc-p2}
        \item $x_i \prec \ell_j^k$ if $x_i$ is the $k$-th literal of clause $c_j$,\label{cc-p7}
        \item $\overline{x}_i \prec \ell_j^k$ if $\overline{x}_i$ is the $k$-th literal of clause $c_j$.\label{cc-p8}
    \end{enumerate}
    
    \begin{claim*}\label{claim:cc-1}
        If $\Phi$ has a fulfilling assignment $\A$, then there is a $\pi$-extending Hamiltonian path in $G'$. 
    \end{claim*}

    \begin{claimproof}
        We start in $s$. Afterwards, we visit for every $i$ the vertex $x_i$ if $x_i$ is set to true in $\A$, otherwise we visit $\overline{x}_i$. Since $\A$ is a fulfilling assignment, there is at least one $k \in [3]$ such that $\ell_1^k$ is now a minimal element. We visit one of these vertices. Afterwards, we visit $a_1$ and then $b_1$. We repeat this procedure for every $i \in [m]$. From $b_m$ we go to the vertex of $x_1$ or $\overline{x}_1$ that is not already visited. Then we visit $y_1$ and then $z_1$. We repeat this procedure for every $i \in [n]$. Now the only remaining vertices are two literal vertices $\ell_j^k$ per clause $c_j$. Since all these vertices are now minimal in $\pi$ and form a clique, we can visit them in an arbitrary order.
    \end{claimproof}

    \begin{claim*}\label{claim:cc-2}
        If there is a $\pi$-extending Hamiltonian path $\cP$ in $G'$, then for every $i \in [n]$ there is at most one of the two vertices $x_i$ and $\overline{x}_i$ that is to the left of $b_m$ in $\cP$.
    \end{claim*}

    \begin{claimproof}
        Consider vertex $y_i$. By \pef{cc-p2}, $y_i$ has to be to the right of $b_m$ and all the vertices $a_j$ with $j \in [m]$ have to be to the left of $b_m$. Therefore, no vertex $a_j$ can be the predecessor of $y_i$ in $\cP$. The only other neighbors of $y_i$ are $z_i$, $x_i$ and $\overline{x}_i$. By \pef{cc-p2}, vertex $z_i$ has to be to the right of $y_i$ in $\cP$. Therefore, one of $x_i$ and $\overline{x}_i$ is the predecessor of $y_i$ in $\cP$ and, thus, this vertex is to the right of $b_m$ in $\cP$.
    \end{claimproof}

    \begin{claim*}\label{claim:cc-3}
        If there is a $\pi$-extending Hamiltonian path $\cP$ in $G'$, then $\Phi$ has a fulfilling assignment.
    \end{claim*}

    \begin{claimproof}
        We set variable $x_i$ to true if and only if vertex $x_i$ is to the left of $b_m$ in $\cP$. We claim that this assignment is fulfilling. Consider the clause $c_j$. Vertex $a_j$ cannot be the first vertex of $\cP$, due to \pef{cc-p1}. Hence, $a_j$ has a predecessor in $\cP$. Due to \pef{cc-p2}, no vertex $y_i$ can be the predecessor since these vertices have to be to the right of $a_j$ in $\cP$. The constraints \pef{cc-p2} also imply that either vertex $b_j$ or vertex $b_i$ has to be between $a_i$ and $a_j$ in $\cP$. Hence, no vertex $a_i$ can be the predecessor of $a_j$. Furthermore, $b_j$ must be to the right of $a_j$. Therefore, the only neighbors of $a_j$ that can be the predecessor are the vertices $\ell_j^1$, $\ell_j^2$, and $\ell_j^3$. However, these vertices can only be taken if their corresponding variable vertex has already been taken, due to \pef{cc-p7} and \pef{cc-p8}. As $a_i$ is to the left of $b_m$, the variable vertex is also to the left of $b_m$. Since by \cref{claim:cc-2}, only one of the two vertices is to the left of $b_m$ in $\cP$, the variable value was set exactly in the way such that clause $c_j$ is fulfilled.
    \end{claimproof}
    
    \Cref{claim:cc-1} and \cref{claim:cc-3} prove the correctness of the reduction. It is easy to see that the graph $G'$ and the partial order $\pi$ can be computed in polynomial time.
\end{proof}

As we have mentioned in \cref{sec:clique-algorithms},  POHPP is \W-hard when parameterized by \param{distance to block}. Here, we strengthen this result by showing that it is even \W-hard when parameterized by \param{distance to clique}.

\begin{theorem}\label{thm:w1-d2c}
     POHPP is \W-hard parameterized by \param{distance to clique}.
\end{theorem}

\begin{proof}
As for \cref{thm:w1-d2p}, we use a reduction from \textsc{Multicolored Clique}. Let $G$ be an instance of MCP with color classes $V_1 \cup \dots \cup V_k$ where $V_i = \{v^i_1, \dots, v^i_q\}$. The construction works similar as for \cref{thm:w1-d2p}. The main difference is how we encode the selection of a vertex from a color class. While in the proof of \cref{thm:w1-d2p} the representative of the selected vertex was visited, here we visit all representatives except from the one of the selected vertex. 

We construct a graph $G'$ using the following gadgets:

    \begin{description}
    \item[Selection Gadget] For every color class $i \in [k]$, we have a clique $X^i$ that contains for every $v^i_p \in V_i$ a vertex $x^i_p$.
    \item[Verification Gadget] For each pair $i,j \in [k]$ with $i < j$, we have a clique $W^{i,j}$ containing for every edge $v^i_pv^j_r \in E(G)$ a vertex $w^{i,j}_{p,r}$.
    \end{description}

    Next we describe how these gadgets are connected to each other (see \cref{fig:w-clique}). 
    The union of the selection gadgets and the verification gadgets forms one large clique. We order the subcliques described in the selection and verification gadgets as follows: \[X^1, X^2, \dots, X^k, W^{1,2}, W^{1,3}, \dots, W^{1,k}, W^{2,3}, W^{2,4}, \dots, W^{k-1,k}.\] We have the following additional vertices in $G'$. First, we have $s^1$, $\hat{s}_1$, and $\hat{t}_1$ that are adjacent to all vertices in $X^1$. For every $i$ with $2 \leq i \leq k-1$, we have vertices $s^i$, $\hat{s}^i$ and $\hat{t}^i$. For $i = k$, we have only $s^i$ and $\hat{s}^i$. Vertex $s^i$ is adjacent to all vertices in $X^i$ and $X^{i-1}$ and the vertices $\hat{s}^i$ and $\hat{t}^i$ are adjacent to all vertices in $X^{i}$. Furthermore, $\hat{s}^i$ is adjacent to $\hat{t}^{i-1}$. For all $i,j \in [k]$ with $i < j$, there are vertices $c^{i,j}$ and $\hat{c}^{i,j}$ that are adjacent to all vertices in $W^{i,j}$ and to all vertices in the clique to left of $W^{i,j}$ in the ordering described above. Finally, we have a vertex $z$ that is adjacent to all vertices in $W^{k-1,k}$ and to $\hat{s}^1$. Observe that the graph $G'$ without the selection and verification gadgets contains $3k + 2 \binom{k}{2} + 1$ vertices. Therefore, the \param{distance to clique} of $G'$ is $\O(k^2)$. 

\begin{figure}
    \centering
    \begin{tikzpicture}[xscale=1.1, clique/.style={draw, rounded corners, minimum size=1cm}]
        \node[clique] (X1) at (0,0) {$X^1$};
        \node[clique] (X2) at (1,0) {$X^2$};
        \node at (2,0) {$\dots$};
        \node[clique] (Xk) at (3,0) {$X^k$};
        \node[clique] (W12) at (4,0) {$W^{1,2}$};
        \node[clique] (W13) at (5,0) {$W^{1,3}$};
        \node at (6,0) {$\dots$};
        \node[clique] (W1k) at (7,0) {$W^{1,k}$};
        \node[clique] (W23) at (8,0) {$W^{2,3}$};
        \node at (8.95,0) {$\dots$};
        \node[clique] (Wk-1k) at (10,0) {$W^{k-1,k}$};

        \node[vertex, label=90:{$s^1$}] (s1) at (-0.5,1) {};
        \node[vertex, label=90:{$s^2$}] (s2) at (0.5,1) {};
        \node[vertex, label=180:{$\hat{t}^1$}] (t2) at (0,2) {};
        \node[vertex, label=0:{$\hat{s}^2$}] (u2) at (1,2) {};
        \node[vertex, label=90:{$c^{1,2}$}] (c12) at (3.5,1) {};
        \node[vertex, label=180:{$\hat{c}^{1,2}$}] (d12) at (3.5,2) {};
        \node[vertex, label=90:{$c^{1,3}$}] (c13) at (4.5,1) {};
        \node[vertex, label=0:{$\hat{c}^{1,3}$}] (d13) at (4.5,2) {};
        \node[vertex, label=90:{$c^{2,3}$}] (c23) at (7.5,1) {};
        \node[vertex, label=0:{$\hat{c}^{2,3}$}] (d23) at (7.5,2) {};
        \node[vertex, label=0:{$z$}] (t1) at (10,-0.825) {};
        \node[vertex, label=180:{$\hat{s}^1$}] (u1) at (0,-0.825) {};

        \draw (s1) -- (X1);
        \draw (s2) -- (X1);
        \draw (s2) -- (X2);
        \draw (t2) -- (X1);
        \draw (t2) -- (u2);
        \draw (u2) -- (X2);
        \draw (d12) -- (Xk.110);
        \draw (c12) -- (Xk);
        \draw (c12) -- (W12);
        \draw (d12) -- (W12.70);
        \draw (d13) -- (W12.110);
        \draw (c13) -- (W12);
        \draw (c13) -- (W13);
        \draw (d13) -- (W13.70);
        \draw (d23) -- (W1k.110);
        \draw (c23) -- (W1k);
        \draw (c23) -- (W23);
        \draw (d23) -- (W23.70);

        \draw (t1) -- (Wk-1k);
        \draw (t1) -- (u1);
        \draw (u1) -- (X1);
    \end{tikzpicture}
    \caption{Construction of the proofs of \cref{thm:w1-d2c,thm:w1-d2cm}. If a vertex is adjacent to a box, then the vertex is adjacent to all vertices in that box. In the proof of \cref{thm:w1-d2c}, all the boxes are pairwise adjacent while in the proof of \cref{thm:w1-d2cm} they all are pairwise non-adjacent.}
    \label{fig:w-clique}
\end{figure}
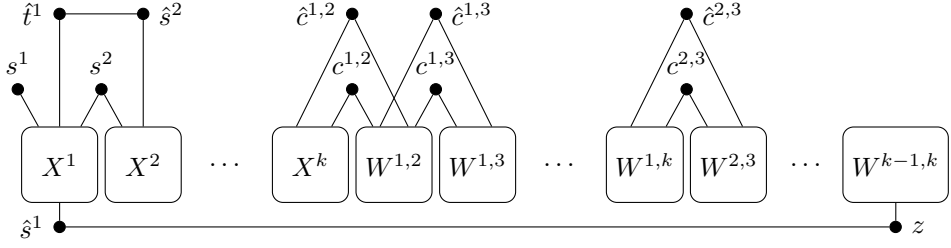

    The partial order $\pi$ is the reflexive and transitive closure of the following constraints.
    \begin{enumerate}[(P1)]
        \item $s^1 \prec v$ for all $v \in V(G') \setminus \{s^1\}$,\label{d2c:p1}
        \item $s^1 \prec s^2 \prec \dots \prec s^k$,\label{d2c:p2}
        \item $z \prec \hat{s}^1 \prec \hat{t}^1 \prec \dots \prec \hat{s}^k$,\label{d2c:p3}
        \item $c^{i,j} \prec c^{i',j'}$ for all $i,j,i',j' \in [k]$ with $i < i'$ or $i = i'$ and $j < j'$,\label{d2c:p4}
        \item $c^{k-1,k} \prec z$, \label{d2c:p5}
        \item $c^{i,j} \prec w^{i,j}_{p,r}$ for all $i,j \in [k]$ with $i < j$ and all $w^{i,j}_{p,r} \in W^{i,j}$,\label{d2c:p6}
        \item $x^i_p \prec w^{i,j}_{r,s}$ for all $i,j \in [k]$ with $i < j$ and all $w^{i,j}_{r,s} \in W^{i,j}$ with $p \neq r$,\label{d2c:p7}
        \item $x^j_p \prec w^{i,j}_{r,s}$ for all $i,j \in [k]$ with $i < j$ and all $w^{i,j}_{r,s} \in W^{i,j}$ with $p \neq s$.\label{d2c:p8}
    \end{enumerate}

    \begin{claim*}\label{claim:clique-direction1}
        If there is a multicolored clique $\{v^1_{p_1}, \dots, v^k_{p_k}\}$ in $G$, then there is a $\pi$-extending Hamiltonian path in $G'$.
    \end{claim*}

    \begin{claimproof}
        First, we visit $s^1$ and then all the vertices in $X^1$ except for the vertex $x^1_{p_1}$. We repeat this for all $i \in [k]$. 
        Next we visit $c^{1,2}$. Now we visit $w^{1,2}_{p_1, p_2}$. Note that this is possible since all the vertices that have to be to the left of $w^{1,2}_{p_1, p_2}$ by \pef{d2c:p7} and \pef{d2c:p8} are already visited. Next we visit $c^{1,3}$ and afterwards $w^{1,3}_{p_1, p_3}$ which is possible for the same reason as mentioned above. We repeat this procedure until we reach $w^{p_1,p_k}$. Next, we visit $c^{2,3}$. The same procedure works until we reach $w^{k-1,k}_{p_{k-1},p_k}$. Next we visit $z$. Now, starting with $i=1$, we visit for all $i \in [k]$ the vertices $\hat{s}^i$, followed by $x^i_{p_i}$ and $\hat{t}^i$ until we reach $x^k_{p_k}$. Then, we visit $\hat{c}^{1,2}$ followed by all remaining vertices $w^{1,2}_{p,r} \in W^{1,2}$ which is possible as all the left vertices in the constraints \pef{d2c:p7} and \pef{d2c:p8} have already been visited. We repeat this for all $i,j$ with $i<j$ in the lexicographic order and finally end with some vertex in $W^{k-1,k}$.
    \end{claimproof}

    It remains to show that a $\pi$-extending Hamiltonian path in $G'$ implies the existence of a multicolored clique in $G$. So assume there is a $\pi$-extending Hamiltonian path $\cP$ in $G'$.

    \begin{claim*}\label{claim:clique-selection1}
        For any $i \in [k]$, there is at least one vertex in $X^i$ that is to the right of $z$ in $\cP$.
    \end{claim*}

    \begin{claimproof}
        By \pef{d2c:p3}, $\hat{s}^i$ is to the right of $z$. One of the two neighbors of $\hat{s}^i$ in $\cP$ has to be in $X^i$. This vertex is to the right of $z$.
    \end{claimproof}

    \begin{claim*}\label{claim:clique-verification1}
        For any $i,j \in [k]$ with $i < j$, there is exactly one vertex of $W^{i,j}$ that is to the left of $z$ in $\cP$.
    \end{claim*}

    \begin{claimproof}
        First assume for contradiction that there are two vertices $w^{i,j}_{p,r}$ and $w^{i,j}_{p',r'}$ to the left of $z$ in $\cP$. It holds that $p \neq p'$ or $r \neq r'$. We assume that $p \neq p'$, the case that $r \neq r'$ works analogously. Due to \pef{d2c:p7}, $x^i_p$ has to be to the left of $w^{i,j}_{p',r'}$ and $x^i_{p'}$ has to be to the left of $w^{i,j}_{p,r}$. All other vertices of $X^i$ have to be to the left of both $w^{i,j}_{p,r}$ and $w^{i,j}_{p',r'}$. Therefore, all vertices of $X^i$ are to the left of $z$ in $\cP$, contradicting \cref{claim:clique-selection1}.

        Now we show that there is also at least one vertex of $W^{i,j}$ to the left of $z$ in $\cP$.
        Consider vertex $c^{a,b}$ with $c^{a,b} \neq c^{1,2}$. Due to \pef{d2c:p1}, \pef{d2c:p4}, and \pef{d2c:p5}, vertex $c^{a,b}$ is not the start vertex of $\cP$ and $c^{a,b}$ is to the left of \(z\) in $\cP$. Thus, it has two neighbors in $\cP$. These neighbors have to be in the two $W$-cliques to which $c^{a,b}$ is adjacent. As we have observed above, there is at most one vertex of any of those cliques to the left of $z$. Therefore, both cliques contain a neighbor of $c^{a,b}$ in $\cP$ and these neighbors are to the left of $z$. As every clique $W^{i,j}$ is adjacent to a vertex $c^{a,b} \neq c^{1,2}$, it follows that every of those cliques contains a vertex that is to the left of $z$.
    \end{claimproof}

    \begin{claim*}\label{claim:clique-selection2}
        For any $i \in [k]$, there is exactly one vertex $x^i_{p_i}$ that is to the right of $z$ in $\cP$.
    \end{claim*}

    \begin{claimproof}
        Due to \cref{claim:clique-selection1}, it remains to show that there is at most one such vertex. As we have seen in \cref{claim:clique-verification1}, there is a vertex in $w^{i,i+1}_{p_i,p_{i+1}}$ or a vertex $w^{i-1,i}_{r_{i-1},r_i}$ to the left of $z$ in $\cP$. The constraints \pef{d2c:p7} or \pef{d2c:p8} imply that all vertices but $x^i_{p_i}$ or $x^i_{r_i}$, respectively, are to the left of~$z$.
    \end{claimproof}
       
    \begin{claim*}\label{claim:clique-verification2}
        Let the $p_i$ be chosen as in \cref{claim:clique-selection2}. The set $\{v^1_{p_1}, \dots, v^k_{p_k}\}$ forms a multicolored clique in $G$.
    \end{claim*}

    \begin{claimproof}
        Let $i,j \in [k]$ with $i < j$. Due to \cref{claim:clique-verification1}, there is a vertex $w^{i,j}_{a,b}$ to the left of $z$ in $\cP$. By \pef{d2c:p7} and \pef{d2c:p8}, all vertices $x^i_p$ and $x^j_r$ with $p \neq a$ and $r \neq b$ have to be to the left of $w^{i,j}_{a,b}$ in $\cP$. Therefore, \cref{claim:clique-selection2} implies that $a = p_i$ and $b = p_j$. By construction of $G'$, this implies that the edge $v^i_{p_i}v^j_{p_j}$ exists in $G$. 
    \end{claimproof}

    \Cref{claim:clique-direction1} and \cref{claim:clique-verification2} prove that $(G', \pi)$ is a proper \FPT{} reduction from MCP to POHPP parameterized by the \param{distance to clique}. This finalizes the proof.
\end{proof}

Similar as for \cref{thm:w1-lfm}, we can adapt the proof of \cref{thm:w1-d2c} to show that  POHPP is also \W-hard for the \param{twin-cover number}, respectively for the \param{distance to cluster modules}. To this end, we delete all the edges in $G'$ between different gadgets. Observe that they have not been used in the proof.

\begin{theorem*}\label{thm:w1-d2cm}
     POHPP is \W-hard when parameterized by \param{twin-cover number}.
\end{theorem*}

\subsection{Algorithms}\label{sec:clique-algorithms}

Now we focus on  POHPP problem for the \param{distance to block}. First we observe that  POHPP can be solved in linear time if the graph is a block graph, i.e., its \param{distance to block} is~0. 

\begin{observation*}\label{obs:block}
     POHPP can be solved on block graphs in $\O(n + m + |\pi|)$ time, where $n$ is the number of vertices and $m$ is the number of edges of the given graph.
\end{observation*}

\begin{proof}
    If the given graph is a clique, then any linear extension of $\pi$ is a solution of  POHPP. As was shown in \cite[Theorem~5.1]{beisegel2024computing}, a linear-time algorithm that solves POHPP on the blocks of a graph implies a linear-time algorithm that solves the problem on the whole graph. Therefore, we can solve  POHPP in $\O(n + m + |\pi|)$ time on block graphs.
\end{proof}

Next, we consider the \param{edge distance to block}.

\begin{theorem}\label{thm:block-edge}
     POHPP can be solved in $k \cdot k! \cdot 2^k \cdot n^{\O(1)} + k^{2k}$ time on an $n$-vertex graph of \param{edge distance to block}~$k$.
\end{theorem}

\begin{proof}
   Due to Dumas et al.~\cite{dumas2024polynomial}, there is a $k^{2k} + n^{\O(1)}$ time algorithm that computes a set $F \subseteq E(G)$ of size $k$ such that $G - F$ is a block graph. Fix such a set $F$.

   Assume there is a $\pi$-extending Hamiltonian path of $G$. Then this path possibly uses some or all of the edges in $F$. We encode the way the path uses these edges by so-called edge choices. An \emph{edge choice} consists of a start vertex $u_0$ of the Hamiltonian path and an ordered subset $\hat{F}$ of $F$ such that for all edges in $\hat{F}$ one of the two directions in which the edge can be traverse is fixed. There are $\leq k \cdot k!$ ordered subsets of $F$, at most $2^k$ choices of the directions and $n$ choices for the start vertex.
   
   Therefore, there are at most $k \cdot k! \cdot 2^k \cdot n$ many edge choices.  In the following, we will call an edge choice \emph{valid} if there is a $\pi$-extending Hamiltonian path starting in $u_0$ such that all edges of $F$ that are part of this path are in $\hat{F}$ and are visited following the order that is fixed in that edge choice. 
   
   We first check whether the ordering of the vertices incident to $\hat{F}$ that is implied by our edge choice is a valid subordering of a linear extension of $\pi$. If this is not the case, we directly reject the edge choice. Otherwise, observe that vertices can be incident to more than one edge in an edge choice $\hat{F}$. However, in every valid edge choice, there are at most two edges of $\hat{F}$ incident to a particular vertex and these edges are consecutive in the ordering of $\hat{F}$. If this is the case, then we normalize the edge choice as follows. A sequence of edges in $\hat{F}$ where consecutive edges share an endpoint is replaced by a new edge going from the first vertex $u$ of this sequence to the last vertex $v$ of the sequence. The partial order is updated accordingly, i.e., all predecessors of some vertex in the sequence that are not part of the sequence become predecessors of $u$ and all successors that are not part of the sequence become successors of $v$.

    This normalization step produces a sequence of pairwise different vertices $(u_1, \dots, u_\ell)$ where for every \emph{even} number $i \in [\ell]$ vertices $u_{i-1}$ and $u_{i}$ are connected via an edge of $\hat{F}$ and $u_{i}$ and $u_{i+1}$ are not connected via an edge of $\hat{F}$. Note that it is possible that the start vertex $u_0$ is equal to $u_1$.
    
    Now, the task of the algorithm is to fill the remaining vertices into the gaps between the $u_i$ or into the gap after $u_\ell$. Let $i \in [\ell] \cup \{0\}$ be even. Consider the block-cut tree $\T$ of $G - F$. For any path $P$ in $G - F$, we define the projection of $P$ into $\T$ as the path in $\T$ that contains for every vertex in $P$ either the block in which it lies if there is a unique one or the vertex itself if the vertex is a cut vertex. As $\T$ is a tree, any path between $u_i$ and $u_{i+1}$ in $G - F$ has the same projection into $\T$. We first compute these projections for every pair $(u_i, u_{i+1})$ where $i$ is even. If the path projection contains cut vertices, then we label these cut vertices with label $i$. Note that we can reject the edge choice if we have to relabel a cut vertex since a cut vertex cannot be used in two different projections. After this process we label all the unlabeled vertices -- in particular all the vertices that are not cut vertices -- with label $-1$. 

    We now traverse the vertices $u_i$ in increasing order. If $i$ is odd, then we directly go to $u_{i+1}$. Otherwise, we follow the projection between $u_i$ and $u_{i+1}$. Whenever we visit a block, we take all vertices that are possible due to $\pi$ except for those vertices that have a label that is larger than $i$. Note that this might include cut vertices to blocks that we do not enter directly afterwards. Vertices that have label exactly $i$ are taken only if there is no other unvisited vertex in that block that can be taken. This is because we have to leave the block when we visit a vertex with label $i$. If we get stuck during this process, we reject that edge choice. Otherwise, we eventually reach vertex $u_\ell$. If the remaining vertices do not induce a connected graph, then we again reject the edge choice. Otherwise, we use the algorithm of \cref{obs:block} to check whether we can traverse these vertices starting in $u_\ell$. 
    
    It is obvious that a Hamiltonian path constructed by the algorithm is $\pi$-extending. It remains to show that the algorithm only fails to find a Hamiltonian path for some edge choice if there is no valid Hamiltonian path for that choice. To this end, assume that we reach a point where our algorithm gets stuck. This means there is no vertex minimal in the remainder of $\pi$ that is adjacent to the last visited vertex. In particular, either $u_{i+1}$ or the next cut vertex on the way to $u_{i+1}$ cannot be taken. Assume for contradiction that there is a $\pi$-extending Hamiltonian path $\cP$ in $G$ following our edge choice. Obviously, $\cP$ traverses the blocks of $G$ and the $u_i$'s in the same order as we have done. However, there must be a vertex $x$ that is traversed in $\cP$ during a block visit where our algorithm did not visit vertex $x$. Let $x$ be the leftmost vertex in $\cP$ that fulfills this property. There are two options why vertex $x$ was not visited by our algorithm. If its label did not fit, then $x$ would also not have been visited in this block in $\cP$  (since otherwise the path $\cP$ would not fit to our edge choice). The only other option is an unvisited vertex that is forced to be to the left of $x$ by $\pi$. However, this vertex would have been visited in $\cP$ contradicting the choice of $x$. Hence, $\cP$ cannot exist.
\end{proof}

We can use the algorithm given in \cref{thm:block-edge} to develop an \XP{} algorithm for the parameter \param{(vertex) distance to block}.

\begin{theorem}
    If the \param{distance to block} of a graph $G$ with $n$ vertices is $k$, then  POHPP can be solved in $n^{\O(k)}$ time.
\end{theorem}

\begin{proof}
    First note that we can find a vertex set $W$ such that $G - W$ is a block graph and $|W| = k$ in time $n^{\O(k)}$ by enumerating all vertex subsets of size $k$. Fix one of these sets $W$. For every vertex in $W$, we choose one predecessor and one successor (or we decide that the vertex is the first or the last vertex of our Hamiltonian path). Furthermore, we fix an ordering of the vertices in $W$. We call these decisions a \emph{vertex choice}. There are $n^{\O(k)}$ many vertex choices. 
    
    For every of those vertex choices, we construct the graph $G'$ as follows. We delete all vertices of $W$ from $G$. Consider the pairs of predecessors and successors of the vertices in $W$. If they form a path starting and ending in vertices of $G - W$, we add an edge to $G'$ connecting the first vertex $s$ of this path with the last vertex $t$ of the path. We update the partial order $\pi$ accordingly, i.e., for any inner vertex on this path, we make its predecessors in $\pi$ to predecessors of $s$ and its successors in $\pi$ to successors of $t$. If we have chosen some vertex of $W$ as start vertex of the Hamiltonian path, then we add constraints to the partial order that makes $t$ (the first successor in $G-W$) the start vertex. Equivalently, if a vertex of $W$ is chosen to be the end vertex of the Hamiltonian path, then we make the vertex $s$ (the last predecessor in $G-W$) the end vertex of the partial order. We call the resulting updated partial order $\pi'$.
    
   It is easy to see that a $\pi$-extending Hamiltonian path of $G$ that follows our vertex choice directly maps to a $\pi'$-extending Hamiltonian path of $G'$ that traverses all the added edges in the order that is implied by the ordering of the vertices in $W$ and by the chosen predecessors and successors. Therefore, we can use the subroutine of \cref{thm:block-edge} that checks for the validity of an edge choice to decide the validity of our vertex choice here. This takes $n^{\O(1)}$ time. Thus, checking all the vertex choices needs $n^{\O(k)}$ time in total.
\end{proof}

As we have seen in \cref{thm:w1-d2c}, there is no \FPT{} algorithm for \param{distance to cluster modules}. However, we can give such an algorithm for \param{distance to clique modules}.

\begin{theorem}
    Given an $n$-vertex graph $G$ with \param{distance to clique module}~$k$,  POHPP can be solved in time $k! \cdot n^{\O(1)}$.
\end{theorem}

\begin{proof}
     First note that \param{distance to clique module} and a corresponding set $W$ can be computed in polynomial time since the largest clique module is equivalent to the largest set of vertices with the same closed neighborhood. Let $W$ be a set of $k$ vertices such that $C = G - W$ is a clique module of $G$. The algorithm considers all \(k!\) orderings $\rho = (v_1, \dots, v_k)$ of $W$ that fulfill the constraints of $\pi$. We call these orderings \emph{choices}.
    
    \subparagraph*{Removing secluded vertices} We say that a vertex of $W$ is \emph{secluded} if it is not adjacent to the vertices in $C$. We observe that the predecessor and the successor of a secluded vertex in an Hamiltonian path has to be an element of $W$. Let $(v_i, v_{i+1}, \dots, v_j)$ be a subsequence of $\rho$ such that $v_i$ and $v_j$ are not secluded but $v_{i+1}, \dots, v_{j-1}$ are secluded. We call $v_i$ and $v_j$ \emph{frontier vertices}. 
    Note that secluded vertices are defined with regard to \(W\), while a vertex is a frontier vertex for a fixed vertex choice \(\rho\).
    If some consecutive vertices in that subsequence are not adjacent in $G$, then a Hamiltonian path cannot follow that ordering of $W$. Thus, we can reject that choice. Otherwise, we contract this sequence to an edge $v_iv_j$. If there is a constraint $x \prec_\pi v_q$ with $i \leq q \leq j$, then the vertex $x$ has to be visited before $v_i$. Hence, we add the constraint $x \prec v_i$ to $\pi$. Similarly, if there is a constraint $v_q \prec_\pi y$ with $i \leq q \leq j$, then $y$ has to be visited after $v_j$ and we add the constraint $v_j \prec y$ to $\pi$. Let $\sigma = (w_1, \dots, w_{k'})$ be the ordering of the subset of non-secluded vertices $W' \subseteq W$ that results from that normalization. We define $G' :=  G[W' \cup C]$ and we call the updated partial order $\pi'$. The following claim is a direct consequence of the explanations above.

    \begin{claim*}\label{claim:secluded}
        If there is a $\pi$-extending path of $G$ following the ordering $\rho$ of $W$, then there is a $\pip$-extending path of $G'$ following the ordering $\sigma$ of $W'$.
    \end{claim*}

    \subparagraph{Adding clique vertices to $\sigma$} We now describe how the algorithm constructs an Hamiltonian path following the ordering $\sigma$ of $W'$. To this end, we define for every vertex $x \in C$ the values $\ell(x) = \max (\{i \mid w_i \prec_\pip x\} \cup \{0\})$ and $r(x) = \min (\{i \mid x \prec_\pip w_i \} \cup \{k'+1\})$. That is $\ell(x)$ is the last \(w_i\in \sigma\) such that \(w_i\) has to be left of \(x\), while \(r(x)\) is the first \(w_i\in \sigma\) that has to be right of \(x\).

    \begin{claim*}\label{claim:lr}
        Let $x,y \in C$. It holds that $\ell(x) < r(x)$ and $\ell(y) < r(y)$. If $x \prec_\pip y$, then $\ell(x) \leq \ell(y)$ and $r(x) \leq r(y)$.
    \end{claim*}

    \begin{claimproof}
        Let $(a,b) = (\ell(x), r(x))$. Assume that $a \neq 0$ and $b \neq k'+1$ (otherwise, $\ell(x) < r(x)$ trivially holds ). Since $w_a \prec_\pip x \prec_\pip w_b$, it holds that $w_a \prec_\pip w_b$. As $\rho$ fulfills the constraints of $\pi$, $a < b$ is true. 
        
        Now assume that $x \prec_\pip y$ and let $(c,d) = (\ell(y), r(y))$. If $a=0$ or $d=k' + 1$, then $a \leq c$ and $b \leq d$ trivially holds. Otherwise, we know that $w_a\prec_\pip x \prec_\pip y$ and $x \prec_\pip y \prec_\pip w_d$. By the definition of the functions $\ell$ and $r$, it holds that $a \leq c$ and $b \leq d$.
    \end{claimproof}

    For every pair $(i,j)$ with $0 \leq i < j \leq k' + 1$, we form a bucket $B_{i,j}$ containing those vertices $x$ with $(\ell(x), r(x)) = (i,j)$. The vertices in $B_{i,j}$ are ordered according to $\pip$, i.e., if there are two vertices $x,y \in B_{i,j}$ with $x \prec_\pip y$ then $x$ is to the left of $y$ in $B_{i,j}$. 

    The algorithm tries to fill the vertices of $C$ into the gaps between the $w_i$'s.
    We traverse $\sigma$ starting in $w_1$. Whenever we reach a vertex $w_i$, then we visit all the unvisited vertices $x \in C$ with $r(x) = i$ directly before $w_i$ following their ordering in $\pip$. Note that we can do this by visiting them in increasing order of $\ell(x)$, due to \cref{claim:lr}. We call these vertices \emph{forced vertices} since the partial order $\pip$ force them to be to the left of $w_i$.  If $i \neq 1$ and $w_iw_{i-1} \notin E(G)$, then we have to visit some vertex of $C$ between $w_{i-1}$ and $w_i$. Hence, if no unvisited vertex with $r(x) = i$ exists, we have to choose another vertex for the gap. Let $(a,b)$ be the tuple with minimal $b$ such that $a < i$ and $B_{a,b}$ contains an unvisited vertex. Note that $b > i$ because all vertices with $r(x) \leq i$ have already been visited. If no such tuple $(a,b)$ exists, then we reject the choice $\rho$. Otherwise, we choose $x$ to be the leftmost unvisited vertex in $B_{a,b}$ and visit it between $w_{i-1}$ and $w_i$. We call $x$ an \emph{unforced vertex} since the partial order $\pip$ did not force it to be to the left of $w_i$. We repeat this process until we have visited $w_k$ or we have rejected the choice. We add all the unvisited vertices after $w_k$ following their ordering in $\pip$. We call the resulting ordering $\tau$.

     \begin{claim*}\label{claim:taui}
        Let $\tau_i$ be the subordering that has been produced after $w_i$ was traversed. The ordering $\tau_i$ induces a path in $G$ and forms a prefix of a linear extension of $\pi$.
    \end{claim*}

    \begin{claimproof}
        First, we show that $\tau_i$ induces a path. If consecutive vertices are both in $C$, then they are adjacent by definition. If both of them are in $W'$, then they are adjacent since otherwise the algorithm would have added some vertex of $C$ between them. If one of them is in $C$ and the other is in $W'$, then they are also adjacent since there is no secluded vertex in $W'$. Hence, $\tau_i$ induces a path.

        It remains to show that $\tau_i$ is a prefix of a linear extension of $\pip$. The constraints on vertices of $W'$ are all fulfilled since we have checked this right at the beginning for $\rho$. If a vertex $x \in C$ was added left of some vertex $w_j \in  W'$, then $\ell(x) \leq j$ and, thus, $w_j \not \prec_\pip x$. If it has been added to the right of some vertex $w_j$, then $r(x) > j$ and, thus, $x \not\prec_\pip w_j$. Finally, consider the case that $x \prec_\pip y$ for some $x,y \in C$. Due to \cref{claim:lr}, it holds that $\ell(x) \leq \ell(y)$ and $r(x) \leq r(y)$. Consider the step in the algorithm where $y$ has been added to $\tau_i$. If $y$ is an forced vertex, then $x$ either was already visited (if $r(x) < r(y)$ or $x$ is an unforced vertex) or it has been added in the same step to the left of $y$. If $y$ is an unforced vertex, then $x$ has already been visited at this point since otherwise $x$ would have been chosen instead of $y$. Hence, $x$ is to the left of $y$ in $\tau_i$.
    \end{claimproof}

    This proves that the algorithm works correctly if it returns an ordering $\tau$. It remains to show that it also works correctly if it rejects the choice $\sigma$ of $W'$.

    \begin{claim*}
        If the algorithm rejects an ordering $\rho$ of $W$, then there is no $\pi$-extending Hamiltonian path of $G$ following that ordering of $W$.
    \end{claim*}

    \begin{claimproof}
        Due to \cref{claim:secluded}, it is sufficient to show that there is no $\pip$-extending Hamiltonian path of $G'$ following $\sigma$. Assume for contradiction that such a Hamiltonian path exists. Let $w_i$ be the vertex where we have rejected and let $\tau_i$ be the ordering that was constructed till that step. We will show in the following that there is a $\pip$-extending Hamiltonian path following the decisions of our algorithm, contradicting the fact that our algorithm rejected $\rho$.
        
        First observe that the only option for a rejection is that we had to take an unforced vertex but we did not find a suitable one in the unvisited vertices. Let $\cP$ be a $\pip$-extending Hamiltonian path of $G'$ following $\sigma$ whose common prefix with $\tau_i$ is as long as possible. 
        Let \(x\) be the first vertex on \(\tau_i\) after the common prefix. 
        If there are multiple \(\pip\)-extending paths with this prefix, we choose one where \(x\) is leftmost. 
        Note that this path is well-defined as $\tau_i$ cannot be a prefix of $\cP$ because otherwise the algorithm would have found a suitable unforced vertex. 
        Let $y$ be the vertex after the common prefix in $\cP$.

        Now we consider the different cases for $x$. If $x$ is a vertex $w_j$, then $y$ and all other vertices between $y$ and $w_j$ in $\cP$ have to be unforced vertices, i.e., $r(z) > j$ for all these vertices $z$. We put all these vertices to the right of $w_j$ following their order in $\cP$. We claim that this results in another $\pip$-extending Hamiltonian path $\cP'$ of $G'$. It is obvious that the result is still a Hamiltonian path as the moved vertices are adjacent to all other vertices in $G'$. Furthermore, as observed above, none of the moved vertices is forced by $\pi'$ to be to the left of $w_j$. Therefore, $\cP'$ is a $\pip$-extending Hamiltonian path with a longer common prefix with $\tau_i$; a contradiction to the choice of $\cP$.

        Next, assume $x$ is a forced vertex, i.e., $r(x) = j$ where $w_j$ is the leftmost vertex of $W'$ to the right of $x$ in $\tau_i$. We know that $x$ has to be between $y$ and $w_j$ in $\cP$. We just move $x$ to the front of $y$. Since $x$ and $y$ are not in $W'$, they form universal vertices in $G'$ and, thus, this results in another Hamiltonian path. By \cref{claim:taui}, the ordering $\tau_i$ is a prefix of a linear extension of $\pip$; hence this Hamiltonian path is also $\pip$-extending. Again this contradicts the choice of $\cP$ as the new Hamiltonian path has a longer common prefix with $\tau_i$.

        Finally, assume that $x$ is an unforced vertex, i.e., $r(x) > j$ where $w_j$ is the leftmost vertex of $W'$ to the right of $x$ in $\tau_i$. This implies that $x$ is the single vertex between the non-adjacent vertices $w_{j-1}$ and $w_j$ in $\tau_i$. This further implies that $y$ is also an unforced vertex since there must be a vertex of $C$ between $w_{j-1}$ and $w_j$ and all vertices $z$ with $r(z) = j$ are to the left of $x$ in $\tau_i$ (otherwise $x$ would not have been chosen by the algorithm). Due to the choice of the algorithm, we know that $r(y) \geq r(x)$. Let now $z \in C$ be the rightmost vertex in $\cP$ that is to the left of $x$ and fulfills the condition $r(z) \geq r(x)$. As observed above, \(y\) is such a vertex, and thus \(z\) exists. We swap $x$ and $z$ in $\cP$. It is obvious that the result is still a Hamiltonian path as $x$ and $z$ are universal vertices in $G'$. It remains to show that the path is still $\pip$-extending. Let $z'$ be a vertex between $z$ and $x$ in $\cP$. Vertex $z'$ is not part of $\tau_i$.  Therefore, as $\tau_i$ is a prefix of a linear extension of $\pip$, the constraint $z' \prec x$ is not part of $\pip$. By the choice of $z$, it must hold that $r(z') < r(x) \leq r(z)$. Due to \cref{claim:lr}, the constraint $z \prec z' $ is also not part of $\pip$. Hence, the swap is allowed and results in another $\pip$-extending Hamiltonian path $\cP'$ of $G'$. However, either the common prefix of $\tau_i$ with $\cP'$ is longer than the common prefix with $\cP$ (if $y = z$) or $x$ is further left in \(\cP'\) than in $\cP$. In both cases, this contradicts the choice of $\cP$.
    \end{claimproof}

    \subparagraph*{Adding back the secluded vertices}
    Assume that the algorithm has not rejected a choice $\rho$ and has produced an ordering $\tau$ of $G'$. 
    To get an ordering of $G$, we have to add the deleted secluded vertices again. To do this, we need the following observation.

    \begin{claim*}\label{claim:frontier}
        No vertex of $C$ is put between two consecutive frontier vertices in $\sigma$.
    \end{claim*}

    \begin{claimproof}
        As consecutive frontier vertices $w_i$ and $w_{i+1}$ are adjacent in $G'$, a vertex $x \in C$ would have been added between them only if $r(x) = i+1$. However, all vertices in $C$ that are to the left of $w_{i+1}$ in $\pip$ are also to the left of $w_i$ in $\pip$, due to the normalization step described above. Thus, $r(x) \neq i+1$ for all vertices $x \in C$.
    \end{claimproof}

    Due to this claim, we can add the removed secluded vertices back between their corresponding frontier vertices and get the ordering $\phi$ of $G$. The predecessors of secluded vertices in $\pi$ have been forced by $\pip$ to be to the left of the left frontier vertex. Similarly, their successors in $\pi$ were forced to be to the right of the right frontier vertex. Hence, the ordering $\phi$ induces a $\pi$-extending Hamiltonian path of $G$.

    As the algorithm considers at most $k!$ orderings of $W$ and for every of these orderings it needs polynomial time, the claimed running time holds.
\end{proof}

\section{Open Problems}

The complexity results presented in \cref{sec:treewidth} leave two main open questions. What is the complexity of  POHPP for graphs of \param{treewidth}~3 and for grid graphs of height $h \in \{4,5,6\}$. \Cref{fig:parameters} presents two parameters where the complexity status is not completely clarified: \param{distance to outerplanar} and \param{edge clique cover number}. For both parameters,  POHPP is \W-hard but it is open whether it is also para-\NP-hard. While we are quite optimistic that there is an \XP{} algorithm for \param{distance to outerplanar}, the status for \param{edge clique cover number} is more obscure to us. 

For the \param{vertex cover number}, we have presented a polynomial kernel for  MinPOHPP. One may ask whether we can do the same for the other \FPT{} results given here. We are quite sure that we have a polynomial kernel for  MinPOHPP parameterized by the \param{feedback edge set number}. However, this kernel has turned out to be surprisingly complex. Therefore, we have decided to omit it here. Instead, we will present it in the journal version of the paper.

Finally, one may ask how the complexity of  POHPP behaves if we restrict both the graph and the partial order. It has been shown in \cite{beisegel2024computing} that  POHPP is in \XP{} and \W-hard when parameterized by the partial order's \param{width}. It is open what happens if we parameterize the problem by both the \param{width} of the partial order and some graph width parameter such as \param{treewidth} or \param{distance to clique}.

\bibliography{lit}
\bibliographystyle{plainurl}

\end{document}